\documentclass[a4paper]{article}

\usepackage[latin1]{inputenc}
\usepackage[english]{babel}
\usepackage[T1]{fontenc}
\usepackage{graphicx}
\usepackage{bm}
\usepackage{setspace}

\usepackage{amsmath,amsthm,amssymb}

\usepackage[sc]{mathpazo} % Use the Palatino font
\usepackage[T1]{fontenc} % Use 8-bit encoding that has 256 glyphs
\usepackage{microtype} % Slightly tweak font spacing for aesthetics

\usepackage[hmarginratio=1:1,top=32mm,columnsep=20pt]{geometry} % Document margins

\usepackage[hang, small,labelfont=bf,up,textfont=it,up]{caption} % Custom captions under/above floats in tables or figures
\usepackage{booktabs} % Horizontal rules in tables
\usepackage{float} % Required for tables and figures in the multi-column environment - they need to be placed in specific locations with the [H] (e.g. \begin{table}[H])
\usepackage{hyperref} % For hyperlinks in the PDF

\usepackage{lettrine} % The lettrine is the first enlarged letter at the beginning of the text
\usepackage{paralist} % Used for the compactitem environment which makes bullet points with less space between them

\usepackage{abstract} % Allows abstract customization
\usepackage{titlesec} % Allows customization of titles

\titleformat{\section}[block]{\large\scshape\centering}{\thesection.}{1em}{} % Change the look of the section titles
\titleformat{\subsection}[block]{\large}{\thesubsection.}{1em}{} % Change the look of the section titles

\usepackage{fancyhdr} % Headers and footers
\numberwithin{equation}{section}
\newcommand{\zeroped}[1]{{\scriptscriptstyle #1}}

\newtheorem{Lems}{Lemma}[section]
\newtheorem{prop}{Proposition}[section]
\newtheorem{theor}{Theorem}[section]

\title{\vspace{-20mm}\fontsize{11pt}{10pt}\selectfont\textbf{ EFFECTS OF CONCAVITY ON THE MOTION OF A BODY IMMERSED IN A VLASOV GAS }} % Article title

\author{
 \Large
 \textsc{Francesco Sisti\footnote{ Dip. di Scienze di Base e Applicate per l'Ingegneria, ''Sapienza'' Universit\'a di Roma , Via A. Scarpa 16, 00161 Roma, Italy. Email: francesco.sisti@sbai.uniroma1.it},  \, Costantino Ricciuti\footnote{Dip. di Scienze Statistiche, ''Sapienza'' Universit\'a di Roma , Piazzale Aldo Moro 5, 00185 Roma, Italy. Email: costantino.ricciuti@uniroma1.it}     } \vspace{-5mm}
 }

\date{}

\begin{document}

\maketitle

\begin{abstract}

We consider a body immersed in a perfect gas, moving under the action of a constant force $E$ along the 
$ x $ axis . 
We assume the gas to be described by the mean-field approximation and interacting elastically with the body.
Such a dynamic was studied in \cite{nostro}, \cite{E=0} and \cite{CONVEX}.
In these studies the asymptotic trend showed no sensitivity whatsoever to the shape of the object moving through the gas.
In this work we investigate how a simple concavity in the shape of the body can affect its asymptotic behavior; we thus consider the case of hollow cylinder in three dimensions or a box-like body in two dimensions. 
We study the approach of the body velocity $V (t) $ to the limiting velocity $V_{\infty} $ and prove that, under suitable smallness assumptions, the approach to equilibrium is 
$| V_{\infty}-V(t)| \approx C  t^{-3}  $ both in two or three dimensions, being $C$ a positive constant. This approach is not exponential, as typical in friction problems, and even slower
 than  for the simple disk and the convex body in $ \mathbb{R}^2 $ or  $ \mathbb{R}^3 $.
\end{abstract}

\tableofcontents 
\section{Introduction}

Let us consider a body of mass $ M $ moving with velocity $ V(t) $ along the x axis under the action of a constant force $ E $, immersed in a homogeneous fluid.
Its time evolution is given by
\begin{equation}
\label{intro V = G(V) +E}
M \dot{V}(t) = -G(V) +E 
\end{equation}
where $ G(V) $ represents the friction term.
This term is usually determined on the basis of simple phenomenological considerations and hence often assumed positive and increasing, there is thus only one stationary solution 
$ V_{\infty} $ to the equation
\begin{equation}
G(V_{\infty}) = E,
\end{equation}
and the solution to (\ref{intro V = G(V) +E}) converges exponentially to this limiting velocity.

In \cite{nostro} a model of free gas of light particles elastically interacting with a simple shaped body (a stick in two dimensions or a disk in three dimensions) is studied  and it is proven that the exponential trend to the limiting velocity is not the most general one, on the contrary the asymptotic time behavior in approaching $ V_{\infty} $ is power-law.
More precisely, assuming the initial velocity $ V_0 $ such that $ V_{\infty}-V_0 $ is positive and small, it was proven that:
\begin{equation}
\label{intro trend d+2}
| V_{\infty}-V(t)| \approx \frac{C}{t^{d+2}}\, ,
\end{equation}
where $ d=1 $, 2, 3 is the dimension of the physical space and $ C $ is a constant, depending on the medium and on the shape of the obstacle.
This result, surprising for not being exponential, is due to re-collisions that can occur between gas particles and the body while it is accelerating.
Moreover, as already stressed in \cite{nostro}, these re-collisions can take place after arbitrarily large time, creating a long tail memory which in the end is responsible for the power law decay.
We incidentally remark that the physical model in \cite{nostro}  has been previously introduced in connection with the so-called piston problem (\cite{piston}).
Later articles (\cite{E=0}) studied the problem in absence of the external force, proving the same power of decay.
Similar model (\cite{diffusivo} )  have been studied where a stochastic kind of interaction between the gas and the body is assumed: when a particle of the medium hits the body it is  absorbed and immediately stochastically re-emitted with a Maxwellian distribution centered  around the body velocity, in this case the behaviour was found to be $ O( \frac{1}{ t^{ d+1} } ) $

The physical reason of this different behaviour is due to the fact that  particles are assumed to be re-emitted with a Maxwellian distribution centered  around the body velocity, therefore a large fraction of the emitted particles have a velocity close to that of the body and this makes recollisions more likely.
More recent works (\cite{strauss},\cite{STRAUSS-2} ) studied a mixed  case where it is assumed that some of the particles that collide with a cylinder-like body reflect elastically,  while others reflect stochastically  with some probability distribution $ K $.
Here the rate of approach of the body to equilibrium is $ O(\frac{1}{ t^{ 3+p} })  $ in three dimensions where $p$ can take any value from 0 to 2, depending on $  K $.

% In the domain of pure elastic collision it is clear that, from (\ref{intro trend d+2}) the trend is sensitive to the space dimension where the dynamic is taking place, in particular in three dimensions the behavior is $ t^{-5} $ like.

Now, since the friction term is due to the interaction between the gas and the object, the question arose of whether the trend of the solution had any connection whatsoever with the simple shape chosen for the object in the above mentioned articles, in particular in the elastic interaction models where gas particles bouncing away from the body keep a stronger track of its shape.

This issue was firstly  faced in \cite{CONVEX}, in the domain of elastic collisions, where it was studied the evolution of a general convex body, 
which  was more delicate to handle and deserved its own analysis, indeed it was shown that
the shape itself was responsible for a change in the coefficient appearing in the upper bound for the velocity $ V(t) $; nevertheless it was confirmed the  same power expressed in (\ref{intro trend d+2}) and this was mainly due to an important feature that the convex body  shares with the first case of a simple disk , namely in both cases colliding gas particles bounce away from the body.

In physically realistic situations it would be desirable to deal with a non ideally smooth shape, one that 
could give rise to possible trapping effects between the body and the gas, indeed as stressed before, what leads to the algebraic decay in the evolution are recollisions and they come essentially from the iterated action between gas particles and the object.

In the present work we removed the hypothesis of convexity and studied the case of a body with lateral barriers of finite length, namely a hollow cylinder in three dimensions or a box-like object in two dimensions  (the case of one dimension had clearly no interest in our analysis) which interacts elastically with the body and moves in a homogeneous fluid with velocity $ V(t) $ along the x axis under the action of a constant external force.

The gas is assumed to be made of free particles (see \cite{CERCIGNANI} on Knudsen gas) elastically interacting with the body and it is studied in the mean field approximation, that is the limit in which the mass of the particles constituting the free gas goes to zero, while the number of particles per unit volume diverges, in such a way that the mass density stays finite. We will make explicit use of this condition in section \ref{Section our problem}.

We prove that both in $  d=2 $ and $ d=3 $ dimensions, if the initial velocity of the body is sufficiently close to its limiting velocity $ V_{\infty} $, then for large $ t $, 
\begin{equation}
\label{intro trend 3}
| V_{\infty}-V(t)| \approx \frac{C}{t^{3}}\, ,
\end{equation}
where $ C $ is a constant, depending on the medium and on the shape of the obstacle.

This result is somehow surprising as, according to the quoted literature, it is the first case 
in which the power of decay is $ t^{-3} $ like and doesn't change from two to three dimensions.

Roughly, the reason of this result is that in the long run barriers retain particles responsible for a further  thrust on the disk ( particles otherwise free to escape, see section
 (\ref{SECTION Computation Recollision terms})  ) and, as shown in the main proof, this layer of gas has an effect that is non transient.
 
From this point of view it is no accident that the trend we found is one dimensional like ($ d=1 $) with respect to (\ref{intro trend d+2}) : in one dimension the gas is constrained to stay in front of the body during the whole evolution; in our case, no matter what the dimension, the simple concavity of the body seems to  act as the same topological constraint.

As already mentioned our technique as that of the quoted articles is perturbative in the sense that we take the parameter $ \gamma = V_{\infty}-V_0\, $ finite but sufficiently small, on the other hand in \cite{NUMERICO 1} and \cite{NUMERICO 2} these models have been numerically studied, in particular they computed the dynamic with stochastic interaction carried out in \cite{diffusivo} for a disk subjected to an harmonic force confirming the analytical results (for analytical studies see \cite{E=0} and references quoted therein), here they removed the hypothesis of initial velocity close to $ V_{\infty} $, (though still positive)   showing that this doesn't affect the dynamic.

We chose a simple shape for the object in order to focus the attention on what we considered the first important feature of a concave body, i.e. possible trapping effects between the object and the gas, indeed this feature turned out to yield the surprising results we have been discussing in this brief introduction.

%%%%%%%%%%%%%%%%%%%%%%%%%%%%%%%%%%%%%%%%%%%%%%  section 2   %%%%%%%%%%%%%%%%%%%%%%%%%%%%%%%%%%%%
%%%%%%%%%%%%%%%%%%%%%%%%%%%%%%%%%%%%%%%%%%%%%%  section 2   %%%%%%%%%%%%%%%%%%%%%%%%%%%%%%%%%%%%
%%%%%%%%%%%%%%%%%%%%%%%%%%%%%%%%%%%%%%%%%%%%%%  section 2   %%%%%%%%%%%%%%%%%%%%%%%%%%%%%%%%%%%%
%%%%%%%%%%%%%%%%%%%%%%%%%%%%%%%%%%%%%%%%%%%%%%  section 2   %%%%%%%%%%%%%%%%%%%%%%%%%%%%%%%%%%%%
%%%%%%%%%%%%%%%%%%%%%%%%%%%%%%%%%%%%%%%%%%%%%%  section 2   %%%%%%%%%%%%%%%%%%%%%%%%%%%%%%%%%%%%
%%%%%%%%%%%%%%%%%%%%%%%%%%%%%%%%%%%%%%%%%%%%%%  section 2   %%%%%%%%%%%%%%%%%%%%%%%%%%%%%%%%%%%%
%%%%%%%%%%%%%%%%%%%%%%%%%%%%%%%%%%%%%%%%%%%%%%  section 2   %%%%%%%%%%%%%%%%%%%%%%%%%%%%%%%%%%%%
%%%%%%%%%%%%%%%%%%%%%%%%%%%%%%%%%%%%%%%%%%%%%%  section 2   %%%%%%%%%%%%%%%%%%%%%%%%%%%%%%%%%%%%
\section{The model }

\subsection{Main features}
\label{Section our problem}

We consider a disk of radius $ R $ with lateral barriers of width $ h $, in dimension $ d = 3 $, namely an hollow  cylinder of height $ h $ without frontal base, subjected to a constant external force $E$ along the x-axis.
The thickness of the body is assumed negligible for sake of simplicity, though this assumption is not essential.
The cylinder is constrained to stay with its base orthogonal to the x axis,  with the center moving along the same axis and with the hollow base facing forwards.

The system is immersed in a perfect gas in equilibrium at temperature $T$ and with constant density 
$\rho $, assumed in the mean field approximation. (that is the limit in which the mass of the particles goes to zero,while the number of particles per unit volume diverges, so that the mass density stays finite.)

The presence of the moving body modifies the equilibrium of the gas which starts to evolve according to the free Vlasov equation.
Our aim is to investigate whether and how the body reaches a limiting velocity.

In what follows we will write  a general vector  $ (r_x,r_y,r_z) \in \mathbb{R}^3 $ as $ \bm{r} = (r_x,\bm{r_{\perp}}) $. 

 %%%
%%% 2 DIMENSIONI SICURO? MENTRE RILEGGO DARE OCCHIO SE NO TOGLIERE...  
 %%%
 
We will refer to our cylinder in space as $ C(t) $,
its bottom, namely the disk, as $ D(t) $ and its side as $ S(t) $ (Figure 1).
More precisely:
\begin{subequations}%%%%%%%%%%%%%%%%%%%%%%%%%%%%%%%%%%%%%%%%%%%%%%%%%%%%%%%%%%%%%%%%%%%%%%
\begin{gather}
C(t) =D(t) \cup S(t), \\
D(t) = \{ (x, \bm{x_{\perp} }) \in \mathbb{R}^3 : \quad x= X(t)  ;\quad | \bm{x_{\perp}} |  \leq R \}, \\
S(t) = \{ (x, \bm{x_{\perp} } ) \in \mathbb{R}^3 :  \quad  0 \leq  x- X(t)  \leq h ;\quad  | \bm{x_{\perp}} \rvert = R  \}
\end{gather} 
\end{subequations}%%%%%%%%%%%%%%%%%%%%%%%%%%%%%%%%%%%%%%%%%%%%%%%%%%%%%%%%%%%%%%%%%%%%%%%%%

Where $ X(t) $ is the position of the cylinder base along the x axis.

Let then $ f(\bm{x},\bm{v}, t) , \ (\bm{x},\bm{v}) \in \mathbb{R}^3 \times \mathbb{R}^3 $ 
be the mass density in the phase space of the gas particles.
It evolves according to the free Vlasov equation:
\begin{equation}%%%%%%%%%%%%%%%%%%%%%%%%%%%%%%%%%%%%%%%%%%%%%%%%%%%%%%%%%%%%%%%%%%%%%%
\label{vlasov}
(\partial_t +\bm{v} \cdot \nabla_{\bm{x}}) f(\bm{x},\bm{v}, t) =0, \qquad \bm{x} \notin C(t)
\end{equation}%%%%%%%%%%%%%%%%%%%%%%%%%%%%%%%%%%%%%%%%%%%%%%%%%%%%%%%%%%%%%%%%%%%%%%%%%

The two dimensional version of this body is a box of length $ 2R $ with barriers of width $ h $,
immersed in 2 dimensional gas with  mass density in the phase space 
$ f(\bm{x},\bm{v}, t) , \ (\bm{x},\bm{v}) \in \mathbb{R}^2 \times \mathbb{R}^2 $.  
For sake of concreteness we shall present the work for the three dimensional case, namely for the cylinder. The remaining case $ d = 2 $ follows by the same arguments with obvious modifications \footnote{For $ \bm{r} \in  \mathbb{R}^2 $ of course $ \bm{r_{\perp}}= r_y  $}.

Let $ \dot{X(t)} = V(t) $ be the velocity of the body and $ (x,\bm{x}_{\perp},v_x ,\bm{v}_{\perp}) $ the position and speed of a particle just before the collision with $ C(t) $ at time $ t $. In our model gas particles will be assumed to perform elastic collision with the body.
In  the case of a collision with the base $ D(t) $, 
$ ( x = X(t),|\bm{x}_{\perp}| \leq R)$, defining $ \bm{v}' =(v'_x ,\bm{v}'_{\perp} ) $ 
as the velocity of the particle after the impact we have :
\begin{subequations}
\label{v'}
\begin{gather}%%%%%%%%%%%%%%%%%%%%%%%%%%%%%%%%%%%%%%%%%%%%%%%%%%%%%%%%%%%%%%%%%%%%%%%%%
v'_x = 2V(t)-v_x  \, , \\
\bm{v}'_{\perp} = \bm{v}_{\perp}.
\end{gather}%%%%%%%%%%%%%%%%%%%%%%%%%%%%%%%%%%%%%%%%%%%%%%%%%%%%%%%%%%%%%%%%%%%%%%%%%
\end{subequations}

\begin{figure}

\includegraphics[scale=0.34]{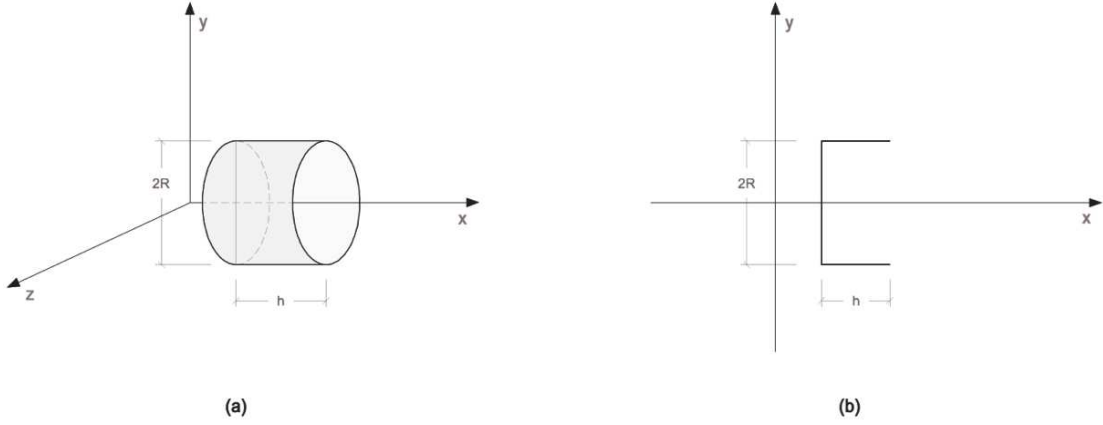}
\caption{ (a) Hollow cylinder of height $h$ ($d=3$), \,  (b) stick with lateral barriers of length $h$ ($d=2$). }
\end{figure}

In  the case of a collision with $ S(t) $, 
$ (  X(t) \leq x \leq X(t)+h , \, |\bm{x}_{\perp}| = R ) $, defining 
$ \widetilde{\bm{v}} =( \widetilde{v_x} , \widetilde{\bm{v}}_{\perp} ) $
as the velocity of the particle after the impact we have :
\begin{subequations}
\label{v tilde}
\begin{gather}
\widetilde{v_x}\,=\,v_x \, , \\
\widetilde{\bm{v}}_{\perp} = \bm{v}_{\perp} - 2 |\bm{v}_{\perp}|\,\cos(\,\theta( \bm{v}_{\perp},\bm{x}_{\perp} ))\hat{\bm{x}}_{\perp},
\end{gather}
\end{subequations}
where $ \theta(\bm{v}_{\perp},\bm{x}_{\perp} )  $ is the angle between $\bm{v}_{\perp} $ and $ \bm{x}_{\perp}$.
Along this work we will make extensive use of the norm condition on $\widetilde{\bm{v}}_{\perp} $ which can be easily deduced from the latter of (\ref{v tilde}):
\begin{equation}
\label{|vperp tilde|= |vperp|}
| \widetilde{\bm{v}}_{\perp} | = | \bm{v}_{\perp} | \, ;
\end{equation}
a  derivation of the foregoing results is carried out in Appendix \ref{appendix collision}.

Together with eq. (\ref{vlasov}) we consider the boundary conditions.
They express conservation of density along trajectories with elastic reflection on $ C(t)$. In particular they are :
\begin{gather}%%%%%%%%%%%%%%%%%%%%%%%%%%%%%%%%%%%%%%%%%%%%%%%%%%%%%%%%%%%%%%%%%%%%%%%%%
\label{boundary conditions}
f_+ (\bm{x},\bm{v}', t) = f_- (\bm{x},\bm{v}, t) \qquad \bm{x} \in D(t) \\
f_+ (\bm{x},\widetilde{\bm{v}}, t) = f_- (\bm{x},\bm{v}, t) \qquad \bm{x} \in S(t)
\end{gather}%%%%%%%%%%%%%%%%%%%%%%%%%%%%%%%%%%%%%%%%%%%%%%%%%%%%%%%%%%%%%%%%%%%%%%%%%
where
\begin{equation}%%%%%%%%%%%%%%%%%%%%%%%%%%%%%%%%%%%%%%%%%%%%%%%%%%%%%%%%%%%%%%%%%%%%%%%%%
f_{\pm}(\bm{x},\bm{v}, t)= \lim _{\epsilon \to 0^+ } f(\bm{x} \pm \epsilon \bm{v}, \bm{v},t \pm \epsilon).
\end{equation}%%%%%%%%%%%%%%%%%%%%%%%%%%%%%%%%%%%%%%%%%%%%%%%%%%%%%%%%%%%%%%%%%%%%%%%%%

Finally we give the initial state of the gas, assumed in thermal equilibrium through the Maxwell- Boltzmann distribution
\begin{equation}%%%%%%%%%%%%%%%%%%%%%%%%%%%%%%%%%%%%%%%%%%%%%%%%%%%%%%%%%%%%%%%%%%%%%%%%%
\label{distribuzione di Maxwell} 
f(\bm{x}, \bm{v},0)= f_0(\bm{v}^2)= \rho ( \frac{\beta}{\pi} ) ^{\frac{3}{2} } e^{- \beta \bm{v}^2} 
\end{equation}%%%%%%%%%%%%%%%%%%%%%%%%%%%%%%%%%%%%%%%%%%%%%%%%%%%%%%%%%%%%%%%%%%%%%%%%%
with $\beta=\frac{1}{kT}$, where $k$ is the Boltzmann constant.

The above equation for the gas is clearly coupled with those of the body immersed in it, which are:
\begin{subequations}%%%%%%%%%%%%%%%%%%%%%%%%%%%%%%%%%%%%%%%%%%%%%%%%%%%%%%%%%%%%%%%%%%%%%%%%%
\label{d/dt V = E- F(t)}
\begin{gather}
\frac{d}{dt} X(t) = V(t), \\
\frac{d}{dt} V(t) = E - F(t), \\
X(0) = 0, \qquad V(0) = V_0, 
\end{gather} 
\end{subequations}%%%%%%%%%%%%%%%%%%%%%%%%%%%%%%%%%%%%%%%%%%%%%%%%%%%%%%%%%%%%%%%%%%%%%%%%%
where $ E $ is the constant external force acting along the x-axis and
\begin{equation}%%%%%%%%%%%%%%%%%%%%%%%%%%%%%%%%%%%%%%%%%%%%%%%%%%%%%%%%%%%%%%%%%%%%%%%%%
\begin{split}
\label{F(t)}
F(t)= 2  \int \limits _ {|\bm{x}_{\perp}| \leq R} d \bm{x}_{\perp} \int \limits _{v_x<V(t)} d\bm{v}
( \, v_x-V(t) \, )^2 f_-(X(t),\bm{x}_{\perp},\bm{v},t) \\
-2  \int \limits _ {|\bm{x}_{\perp}| \leq R} d \bm{x}_{\perp} \int \limits _{v_x>V(t)} d\bm{v}
( \, v_x-V(t) \, )^2 f_-(X(t),\bm{x}_{\perp},\bm{v},t) 
\end{split}
\end{equation}%%%%%%%%%%%%%%%%%%%%%%%%%%%%%%%%%%%%%%%%%%%%%%%%%%%%%%%%%%%%%%%%%%%%%%%%%
is the action of the gas on the disk.

We give here a derivation of eq.(\ref{d/dt V = E- F(t)}) and (\ref{F(t)}),
for sake of simplicity we will denote $ V(t) $ simply as $ V $. \\
Our body, while moving, is subjected to multiples collision with gas particles, if we write its total variation of momentum in an interval $(t,t + \Delta t )$ along the $x$ axis as $ \Delta V(t)  $ and the variation of momentum due to collisions with gas particles as $ \Delta V_{coll}(t)  $, we have that ($M$ is the mass of the body) 
\begin{equation}%%%%%%%%%%%%%%%%%%%%%%%%%%%%%%%%%%%%%%%%%%%%%%%%%%%%%%%%%%%%%%%%%%%%%%%%%
M \Delta V(t) = E \Delta t +  M \Delta V_{coll}(t).
\end{equation}%%%%%%%%%%%%%%%%%%%%%%%%%%%%%%%%%%%%%%%%%%%%%%%%%%%%%%%%%%%%%%%%%%%%%%%%%

After one collision at time $ t $ between a particle of position and speed \\
$ (x,\bm{x}_{\perp},v_x ,\bm{v}_{\perp}) $  and the cylinder base\footnote{ We remember that collisions with $ S(t) $ don't affect momentum along $ x $} $ D(t) $  the change in momentum along x-axis is $ (2m / M ) ( \, v_x -V \, )  $ (see Appendix (\ref{appendix collision})).

The term $ \Delta V_{coll}(t) $ takes into account all the collisions happening during $ \Delta t $, thus:
\begin{equation}%%%%%%%%%%%%%%%%%%%%%%%%%%%%%%%%%%%%%%%%%%%%%%%%%%%%%%%%%%%%%%%%%%%%%%%%%
\label{DP_coll somme in k}
\Delta V_{coll}(t)= \frac{2m}{M}  \sum_{k} ( \, v^k_x -V \, ) +\alpha
\end{equation}%%%%%%%%%%%%%%%%%%%%%%%%%%%%%%%%%%%%%%%%%%%%%%%%%%%%%%%%%%%%%%%%%%%%%%%%%
where $ k $ labels all particles around the cylinder that are hitting the disk $ D(t) $ within 
$ \Delta t $, and $ \alpha $ denotes terms $ o(\Delta t)  $

Let $ \Delta \bm{x}^i \Delta \bm{v}^j $ be a volume of the phase space of measure
 $ \vert \Delta \bm{x}^i \Delta \bm{v}^j \rvert  = \Delta W $, centered at the point $ (\bm{x}^i,\bm{v}^j)$ and
 $\Delta N (\bm{x}^i,\bm{v}^j,t) $ the number of particles contained in it at time $ t $, so that :
\begin{gather}%%%%%%%%%%%%%%%%%%%%%%%%%%%%%%%%%%%%%%%%%%%%%%%%%%%%%%%%%%%%%%%%%%%%%%%%%
\frac{2m}{M} \sum_{k} ( \, v^k_x -V \, )=
\frac{2m}{M} \sum_{i j}( \, v^j_x -V \, ) \Delta N (\bm{x}^i,\bm{v}^j,t)  \notag \\
= \frac{2}{M} \sum_{i j}( \, v^j_x -V \,) \, m \, 
\frac{\Delta N (\bm{x}^i,\bm{v}^j,t) }{\Delta W } \Delta W, \notag
\end{gather}%%%%%%%%%%%%%%%%%%%%%%%%%%%%%%%%%%%%%%%%%%%%%%%%%%%%%%%%%%%%%%%%%%%%%%%%%
where $ i, j $  ranges over positions and velocities that will give rise to collision in 
$ (t, t+ \Delta t) $.
At this point letting $ \Delta W \to 0 $ the mean field approximation,whose meaning was mentioned in the introduction, guarantees the  convergence to a finite mass density, i.e :
\begin{equation}%%%%%%%%%%%%%%%%%%%%%%%%%%%%%%%%%%%%%%%%%%%%%%%%%%%%%%%%%%%%%%%%%%%%%%%%%
 \lim_{ \Delta W \to 0 } m \, \frac{\Delta N (\bm{x}^i,\bm{v}^j,t) }{\Delta W } = f(\bm x, \bm v,t),
\end{equation}%%%%%%%%%%%%%%%%%%%%%%%%%%%%%%%%%%%%%%%%%%%%%%%%%%%%%%%%%%%%%%%%%%%%%%%%%
so that we arrive to:
\begin{equation}%%%%%%%%%%%%%%%%%%%%%%%%%%%%%%%%%%%%%%%%%%%%%%%%%%%%%%%%%%%%%%%%%%%%%%%%%
\label{DP coll integrale}
\Delta V_{coll}(t) =\frac{2}{M} \int \limits _{\Omega(\Delta t)}( \, v_x -V \,)f(\bm x, \bm v,t) d\bm x d\bm v + \alpha
\end{equation}%%%%%%%%%%%%%%%%%%%%%%%%%%%%%%%%%%%%%%%%%%%%%%%%%%%%%%%%%%%%%%%%%%%%%%%%%
where $ \Omega(\Delta t) $ is the $ (\bm x, \bm v)$ region of particles hitting $ D(t) $ in $(t, t+ \Delta t)$;\, for further convenience we split this integral into frontal contribution to recoliision 
$  \Omega^+(\Delta t) $ and backward contribution $ \Omega^-(\Delta t) $ :
\begin{equation}%%%%%%%%%%%%%%%%%%%%%%%%%%%%%%%%%%%%%%%%%%%%%%%%%%%%%%%%%%%%%%%%%%%%%%%%%
\label{DP coll integrale}
\int \limits _{\Omega^+(\Delta t)}\frac{2}{M}( \, v_x -V \,)f(\bm x, \bm v,t) d\bm x d\bm v \,  +
\int \limits _{\Omega^-(\Delta t)}\frac{2}{M}( \, v_x -V \,)f(\bm x, \bm v,t) d\bm x d\bm v.
\end{equation}%%%%%%%%%%%%%%%%%%%%%%%%%%%%%%%%%%%%%%%%%%%%%%%%%%%%%%%%%%%%%%%%%%%%%%%%%

In order for a frontal recollision to happen within an interval of time $(t,t + \Delta t)  $ it is necessary that 
$ V -v_x \geq 0 $, then there must be a time
$   0 \leq \eta -t \leq  \Delta t $ such that the particle and the disk occupies the same position on the x-axis, i.e :
\begin{gather*}%%%%%%%%%%%%%%%%%%%%%%%%%%%%%%%%%%%%%%%%%%%%%%%%%%%%%%%%%%%%%%%%%%%%%%%%%
x + v_x(\eta -t) = X(\eta) = X(t) + V(\eta -t) + \alpha, \text{   \, or }  \\
x- X(t) = (V- v_x) (\eta -t) + \alpha,
\end{gather*}%%%%%%%%%%%%%%%%%%%%%%%%%%%%%%%%%%%%%%%%%%%%%%%%%%%%%%%%%%%%%%%%%%%%%%%%%

finally, among particles starting outside the barriers only those with \\
$ x -X(t)> h $ would be able to enter and hit $ D(t) $ but $ h $ is finite while the first condition implies  $ x -X(t) = O(\Delta t) $,
hence the last condition : $ | \bm{x}_{\perp} | < R $ .

Summarizing it all:
\begin{align}%%%%%%%%%%%%%%%%%%%%%%%%%%%%%%%%%%%%%%%%%%%%%%%%%%%%%%%%%%%%%%%%%%%%%%%%%
\label{OMEGA +}
\Omega & ^+  (\Delta t) = \{ (\bm x, \bm v) \in  \mathbb{R}^3 \times \mathbb{R}^3 : \ V -v_x \geq 0, \\   \notag  & 0 \leq x- X(t) \leq (V- v_x)\Delta t + \alpha , \ \,  | \bm{x}_{\perp} | < R \}.
\end{align}%%%%%%%%%%%%%%%%%%%%%%%%%%%%%%%%%%%%%%%%%%%%%%%%%%%%%%%%%%%%%%%%%%%%%%%%%

Therefore the first contribution is:
\begin{gather*}%%%%%%%%%%%%%%%%%%%%%%%%%%%%%%%%%%%%%%%%%%%%%%%%%%%%%%%%%%%%%%%%%%%%%%%%%
\int \limits _ {v_x<V} d\bm{v} \int \limits _{0}^{(V-v_x)\Delta t +\alpha} d( x-X(t) ) \notag
\int \limits _ {|\bm{x}_{\perp}|< R} d \bm{x}_{\perp} \frac{2}{M} (v_x-V) f(\bm{x},\bm{v},t)  \\
= -\int \limits _ {v_x<V} d\bm{v} \int \limits _ {|\bm{x}_{\perp}|< R} d \bm{x}_{\perp} \frac{2}{M} (v_x-V)^2 f_-( X(t),\bm{x}_{\perp},\bm{v},t) \Delta t \, +  \, \alpha  \\
= - \int \limits _ {|\bm{x}_{\perp}|< R} d \bm{x}_{\perp} \int \limits _ {v_x<V} d\bm{v} \frac{2}{M} (v_x-V)^2 f_-( X(t),\bm{x}_{\perp},\bm{v},t) \Delta t \, +  \, \alpha . 
\end{gather*}%%%%%%%%%%%%%%%%%%%%%%%%%%%%%%%%%%%%%%%%%%%%%%%%%%%%%%%%%%%%%%%%%%%%%%%%%

Regarding backward recollisions, it is necessary that $ V -v_x \leq 0 $ and again there must be a time $   0 \leq \eta -t \leq  \Delta t $ such that 
\begin{gather}%%%%%%%%%%%%%%%%%%%%%%%%%%%%%%%%%%%%%%%%%%%%%%%%%%%%%%%%%%%%%%%%%%%%%%%%%
x- X(t) = -(v_x- V) (\eta -t) + \alpha,
\end{gather}
In this case though there aren't barriers that guarantee condition on $ \bm{x}_{\perp} $ inside (\ref{OMEGA +}) to hold; thus the particle must be on the surface of the disk at the impact time 
$ \eta $,\, i.e ( we write $ (\eta -t) =\delta_t $ remembering that $ \delta_t = \delta_t(x,v_x)  $ ):
\begin{gather}
|\bm{x}_{ \perp} + \bm{v}_{ \perp} \delta_t(x,v_x) |< R \,; \\ 
\delta_t = \frac{x-X(t)}{V- v_x} +\alpha \leq \Delta t;
\end{gather}%%%%%%%%%%%%%%%%%%%%%%%%%%%%%%%%%%%%%%%%%%%%%%%%%%%%%%%%%%%%%%%%%%%%%%%%%

Summarizing it all:
\begin{align}%%%%%%%%%%%%%%%%%%%%%%%%%%%%%%%%%%%%%%%%%%%%%%%%%%%%%%%%%%%%%%%%%%%%%%%%%
\label{OMEGA -}
\Omega & ^-  (\Delta t) = \{ (\bm x, \bm v) \in  \mathbb{R}^3 \times \mathbb{R}^3 : \ V -v_x \leq 0, \\          & \notag  - (v_x- V)\Delta t + \alpha \leq x- X(t) \leq 0 , \ \,  
| \bm{x}_{\perp}+ \bm{v}_{ \perp} \delta_t(x,v_x)  | < R \}
\end{align}%%%%%%%%%%%%%%%%%%%%%%%%%%%%%%%%%%%%%%%%%%%%%%%%%%%%%%%%%%%%%%%%%%%%%%%%%

consequently the second contribution is

\begin{gather*}%%%%%%%%%%%%%%%%%%%%%%%%%%%%%%%%%%%%%%%%%%%%%%%%%%%%%%%%%%%%%%%%%%%%%%%%%
= \int \limits _ {v_x \geq V} d\bm{v} \int \limits_{ -(v_x- V)\Delta t +\alpha}^{0} d( x-X(t) ) 
\int \limits _ {| \bm{x}_{\perp}+ \bm{v}_{ \perp} \delta_t | < R} 
 \frac{2}{M} (v_x-V) f(\bm{x},\bm{v},t) d\bm{x}_{\perp} \\
= \int \limits _ {v_x \geq V} d\bm{v} \int \limits_{ -(v_x- V)\Delta t +\alpha}^{0} d( x-X(t) ) 
\int \limits _ {|\bm{y}_{\perp} | < R} 
\frac{2}{M} (v_x-V) f(x, \bm{y}_{\perp}- \bm{v}_{ \perp} \delta_t ,\bm{v},t) d\bm{y}_{\perp}   \\
= \int \limits _ {v_x \geq V} d\bm{v} \int \limits_{ -(v_x- V)\Delta t +\alpha}^{0} d( x-X(t) ) 
\int \limits _ {|\bm{y}_{\perp} | < R} 
 \frac{2}{M} (v_x-V) f(x, \bm{y}_{\perp},\bm{v},t) + O(\Delta t) \, d\bm{y}_{\perp} \\
 =\int \limits _ {v_x \geq V} d\bm{v} \int \limits_{|\bm{y}_{\perp} | < R} 
 \bigl [ \frac{2}{M} (v_x-V) f(X(t), \bm{y}_{\perp},\bm{v},t) + O(\Delta t) \bigl]  (\,(v_x-V)\Delta t+\alpha ) \, d\bm{y}_{\perp} \\
 = \int  \limits_{|\bm{y}_{\perp} | < R} d\bm{y}_{\perp} \int \limits _ {v_x \geq V} d\bm{v} 
  \frac{2}{M} (v_x-V)^2  f_-(X(t), \bm{y}_{\perp},\bm{v},t)\Delta t  +\alpha
\end{gather*}%%%%%%%%%%%%%%%%%%%%%%%%%%%%%%%%%%%%%%%%%%%%%%%%%%%%%%%%%%%%%%%%%%%%%%%%%

where we performed the change of variable $ \bm{y}_{\perp} = \bm{x}_{\perp}+ \bm{v}_{ \perp} \delta_t $ to shift the dependence on time from the integral region to the integrand and expanded in power series to order $ O(\Delta t) $.
We can finally write  (we set $ M= 1$, $ M $ being an irrelevant constant
\begin{gather}%%%%%%%%%%%%%%%%%%%%%%%%%%%%%%%%%%%%%%%%%%%%%%%%%%%%%%%%%%%%%%%%%%%%%%%%%
\frac{\Delta V(t)}{\Delta t} = E   - \, 2 \int \limits _ {|\bm{x}_{\perp}|< R} d \bm{x}_{\perp} \int \limits _ {v_x<V} d\bm{v} (v_x-V)^2 f_-( X(t),\bm{x}_{\perp},\bm{v},t)  \notag \\ 
+ 2 \int  \limits_{|\bm{x}_{\perp} | < R} d\bm{x}_{\perp} \int \limits _ {v_x \geq V} d\bm{v} 
  (v_x-V)^2  f_-(X(t), \bm{x}_{\perp},\bm{v},t) + O(\Delta t)
\end{gather}%%%%%%%%%%%%%%%%%%%%%%%%%%%%%%%%%%%%%%%%%%%%%%%%%%%%%%%%%%%%%%%%%%%%%%%%%
taking the limit ($ \Delta t \to 0 $ ) concludes the derivation of eq.(\ref{d/dt V = E- F(t)}) and (\ref{F(t)}).

%%%%%%%%%%%%%%%%%%%%%%%%%%%%%%%%%%%%%%%%%%%%%%SECTION  RECOLLISION%%%%%%%%%%%%%%%%%%%%%%%%%%%%%%%
%%%%%%%%%%%%%%%%%%%%%%%%%%%%%%%%%%%%%%%%%%%%%%SECTION  RECOLLISION%%%%%%%%%%%%%%%%%%%%%%%%%%%%%%%
%%%%%%%%%%%%%%%%%%%%%%%%%%%%%%%%%%%%%%%%%%%%%%SECTION  RECOLLISION%%%%%%%%%%%%%%%%%%%%%%%%%%%%%%%
%%%%%%%%%%%%%%%%%%%%%%%%%%%%%%%%%%%%%%%%%%%%%%SECTION  RECOLLISION%%%%%%%%%%%%%%%%%%%%%%%%%%%%%%%
%%%%%%%%%%%%%%%%%%%%%%%%%%%%%%%%%%%%%%%%%%%%%%SECTION  RECOLLISION%%%%%%%%%%%%%%%%%%%%%%%%%%%%%%%
%%%%%%%%%%%%%%%%%%%%%%%%%%%%%%%%%%%%%%%%%%%%%%SECTION  RECOLLISION%%%%%%%%%%%%%%%%%%%%%%%%%%%%%%%

\subsection{Recollision terms}
\label{Section recollision terms}
The aim of our work is to derive the asymptotic behaviour of the body.
Equation (\ref{d/dt V = E- F(t)}) shows that its motion is coupled to that of the gas through the friction term in which $ f(\bm{x}, \bm{v},t) $ is present. The P.d.f.\footnote{Probability density function.} of the gas can be solved by means of characteristics (see for example \cite{BLA} on this).

Indeed let $ \bm{x}(s,t,\bm{x},\bm{v}) $, $ \bm{v}(s,t,\bm{x},\bm{v}) $ be the position and velocity of a particle at time $ s \leq t $, that at time  $ t $ occupies position  $ \bm{x} $ and velocity
$\bm{v} $ ; 
conservation of mass implies that the P.d.f. stays constant along particles trajectories and in particular 
\begin{equation}
f(\bm{x}, \bm{v},t)= f_0( \bm{x}(0,t,\bm{x},\bm{v}) , \bm{v}(0,t,\bm{x},\bm{v}))
\end{equation}
so that the problem of finding the gas distribution reduces to that of tracking the particles trajectories.

Given the evolution of the cylinder $X(t) $ , $ V(t) $, there is a unique backward time evolution leading to the initial position and velocity.  
Such backward evolution is a free motion up to possible collision-times in which the particle hits the body. On these times we keep track of  the particle displacement through condition (\ref{v'}) and (\ref{v tilde}). We proceed in this way until we reach the desired
 $\bm{x}(0,t,\bm{x},\bm{v}) $ ,  $ \bm{v}(0,t,\bm{x},\bm{v})$.
At the end using the initial state of the gas distribution, eq.(\ref{distribuzione di Maxwell}), we obtain
\begin{gather}
\label{F(t) e^-bv^2}
F(t)= 2 \rho ( \frac{\beta}{\pi} )^{\frac{3}{2}} \Bigl [ \int \limits _ {|\bm{x}_{\perp}| \leq R} d \bm{x}_{\perp} \int \limits _{v_x<V(t)} d\bm{v}
( \, v_x-V(t) \, )^2 e^{ -\beta  \bm{v}_{\zeroped{0}}^2 } \notag \\
- \int \limits _ {|\bm{x}_{\perp}| \leq R} d \bm{x}_{\perp} \int \limits _{v_x \geq V(t)} d\bm{v}
( \, v_x-V(t) \, )^2 e^{-\beta  \bm{v}_0^2 } \Bigl ]
\end{gather}
Where $ \bm{v}_{\zeroped{0}} = \bm{v}(0,t,X(t), \bm{x}_{\perp} , \bm{v}) $ or by components:

$\bm{v}_{\zeroped{0}} =  (v_{\zeroped{0}x },\bm{v}_{\zeroped{0\perp} }) = 
( v_{x}(0,t,X(t),\bm{x}_{\perp},\bm{v}), \bm{v}_{\zeroped{\perp}}(0,t,X(t),\bm{x}_{\perp},\bm{v})  )$.
\\
Note that in order to compute $ F(t) $ we need to evaluate $ \bm{v}_{\zeroped{0}} $ and hence to know all the previous history  $\{ X(s), V(s) , s<t \}$.\\
On the other hand, if the light particle goes back without undergoing any collision, then 
$ \bm{v}_{\zeroped{0}} =\bm{v} $ and the friction term is easily computed:
\begin{gather}
\label{F0(V)}
F_{ \zeroped{0} }(V)= A \Bigl [ \int \limits_{v_x<V(t)} dv_x
( \, v_x-V(t) \, )^2 e^{ -\beta  v_x^2 } 
-  \int \limits _{v_x \geq V(t)} dv_x
( \, v_x-V(t) \, )^2 e^{-\beta  v_x^2  } \Bigl ],  \notag \\
\end{gather}
where the constant contains the area of $ D(t) $ and the integral over $ \bm{v}_{\zeroped{\perp} } $ velocity:
\begin{gather}
 A= 2 \rho ( \frac{\beta}{\pi} )^{\frac{3}{2}} \, \pi R^2 \int_{\mathbb{R}^2} d\bm{v}_{\zeroped{\perp} } e^{-\beta \bm{v}_{\zeroped{\perp} }^2} =
2 \rho R^2 \sqrt{\pi \beta} 
\end{gather}

In this case the body moves as it was always immersed in the unperturbed equilibrium state of the gas and its dynamic is no more coupled to that of the gas;
it is thus convenient to split $ F(t) $ in terms that contain recollisions and terms that do not; the following expression serves the purpose:
\begin{gather}
F(t) = F_{ \zeroped{0} }(V) + r^+(t)+r^-(t)
\end{gather}
where:
\begin{gather*}
r^+(t)= 2\rho (\frac{\beta}{\pi})^{\frac{3}{2}}
\int\limits_{ |\bm{x}_{\perp}|\leq R} d\bm{x}_{\perp}\int\limits_{v_x < V(t)} 
d\bm{v}(v_x-V(t))^2  (e^{-\beta \bm{v}_{0}^2}-e^{-\beta \bm{v}^2}), \\  
r^-(t)= 2\rho   (\frac{\beta}{\pi} )^{\frac{3}{2}}
\int\limits_{ |\bm{x}_{\perp}| \leq R} d\bm{x}_{\perp} \int\limits_{v_x \geq V(t)} 
d\bm{v}(v_x-V(t))^2  (e^{-\beta \bm{v}^2}- e^{-\beta \bm{v}_{0}^2}),
\end{gather*}
besides when colliding $ C(t) $ a particle changes in general its components, nevertheless eq.(\ref{v'}b) and (\ref{|vperp tilde|= |vperp|}) imply
\begin{equation}
\bm{v}_{\zeroped{0\perp} }^2 = \bm{v}_{\zeroped{\perp} }^2
\end{equation}
and hence
\begin{equation*}
e^{-\beta \bm{v}_{0}^2} = e^{-\beta( v_{\zeroped{0}x }^2 + \bm{v}_{\zeroped{0\perp} }^2)}=
e^{-\beta \bm{v}_{\zeroped{\perp} }^2}e^{-\beta v_{\zeroped{0}x }^2 }
\end{equation*}
so that we can finally write recollision terms in their general form 
\begin{equation} 
\label{r+}
 r^+(t) =2\rho (\frac{\beta}{\pi})^{\frac{3}{2}}\int\limits_{|\bm{x}_{\perp}|\leq R}
 d\bm{x}_{\perp} \int\limits_{v_x<V(t)} dv_x (v_x-V)^2\int \limits_{\mathbb{R}^2} d\bm{v}_{\perp} e^{-\beta \bm{v}_{\perp}^2} (e^{-\beta v_{0x}^2} -e^{- \beta v_x^2})  \\
\end{equation}
\begin{equation} 
\label{r-}
r^-(t) =2\rho (\frac{\beta}{\pi})^{\frac{3}{2}}\int\limits_{|\bm{x}_{\perp}|\leq R}
d\bm{x}_{\perp} \int\limits_{v_x\geq V(t)} dv_x (v_x-V)^2\int \limits_{\mathbb{R}^2} d\bm{v}_{\perp} e^{-\beta \bm{v}_{\perp}^2} (e^{-\beta v_{x}^2} -e^{- \beta v_{0x}^2}).
\end{equation}

The analysis that has been carried out up to this point makes clear that recollision terms contain the "true" coupling action between the gas and the body. Now, in view of the physical system that will be discussed in the next section, we consider a time evolution with $V(t)>0$ for all $t>0$. In this case $r^+$ and $r^-$ both decelerate the body  with respect to the non recollisional case; i.e:
\begin{equation}
r^{\pm}(t) \geq 0.
\end{equation}
Indeed let $ \tau $ be the first time of collision between $ D(t) $ and a particle with velocity $ \bm{v}_{\zeroped{0}} $ before the impact and  velocity $\bm{v} $ afterwards; concerning $ r^+(t) $, necessary condition for the recollision to happen is
$ v_{\zeroped{0}x}< V(\tau) $, $ v_x < V(t) $.
After the collision
\begin{gather*}
v_x = 2V(\tau)- v_{\zeroped{0}x} =V(\tau)+(V(\tau)- v_{\zeroped{0}x})>V(\tau)
\end{gather*}
so $ v_x > V(\tau) > v_{\zeroped{0}x} $ , on the other hand 
$ v_{\zeroped{0}x}= -v_x + 2V(\tau) > - v_x $ \\
and the resulting inequality ( $ -v_x < v_{\zeroped{0}x} < v_x $) implies 
\begin{equation*}
e^{-\beta v_{0x}^2} \geq e^{-\beta v_x^2 }
\end{equation*}
If more collisions take place, we just iterate the above argument with 
$ \bm{v}_{\zeroped{1}} = \bm{v} $.
Concerning  $ r^-(t)  $ necessary condition for the recollision to happen is 
$ v_{\zeroped{0}x} > V(\tau) $ , $ v_x > V(t) $ similar computation as before gives
$ 0< V(t)< v_x <V(\tau) < v_{\zeroped{0}} $ that implies
\begin{equation*}
e^{-\beta v_{0x}^2} \leq e^{-\beta v_x^2 }.
\end{equation*}
Another inequality that comes directly from here, and from the fact that\\
$ 0 < e^{-v^2} \leq 1 $, is
\begin{gather}
e^{-\beta v_{0x}^2} -e^{- \beta v_x^2} 	\leq 1 \, ; \quad v_x \leq V(t) \\
e^{-\beta v_{x}^2} -e^{- \beta v_{0x}^2} \leq 1 \, ; \quad v_x \geq V(t) .
\end{gather}

%
%
%
%   Taken out 
%
%
%

It is worth remarking that $ F(t) $ reduce to $ F_{ \zeroped{0} }(V) $ also disregarding only collisions with bottom $ D(t) $ but taking into account those with the side,
indeed by virtue of eq.(\ref{v'}) and eq.(\ref{|vperp tilde|= |vperp|})
\begin{gather*}
| \bm{v}_{\zeroped{0}} |^2 = |v_{\zeroped{0}x }|^2 + |\bm{v}_{\zeroped{0\perp} }|^2 =
| \bm{v} |^2.
\end{gather*}

We point out that disregarding recollisions between the gas and the cylinder will uncouple the cylinder dynamic to that of the gas making the solution straightforward.
In this case, in fact, the friction term reduces to $ F_{ \zeroped{0} }(V) $ and in  Appendix \ref{appendix F_0(V)} we prove that it is an odd function and for $V>0 $ it is positive, increasing and convex. 
The cylinder then moves according to the differential equation: 
\begin{equation}
  \dot{V(t)} = E - F_{ \zeroped{0} }(V) 
\end{equation}
; that has a stationary solution $V(t) \equiv V_{\infty}$, with $V_{\infty}>0$, such that
\begin{equation}
\label{F0(Vinf)=E }
F_{ \zeroped{0} }(V_{\infty})=E
\end{equation}
and this is unique because $F_{ \zeroped{0} }(V)$ is monotone, so that we can write 
the equation as
\begin{gather}
\label{d/dt V = F0(V_inf)- F0(t)}  
\frac{d}{dt} V(t) = F_{ \zeroped{0} }(V_{\infty})-F_{ \zeroped{0} }(V) , \\
V(0) = V_0.
\end{gather}
Now, exploiting the properties of $F_{ \zeroped{0} }(V)$, by standard comparison argument (see also (\cite{nostro})  it is straightforward to show that, for  $0<V_0<V_{\infty}$
\begin{gather}
\label{AUTONOM. gamma e^ < V_inf - V < gamma e^}
\gamma e^{-C_- t}\leq  V_{\infty}-V(t) \leq \gamma e^{-C_+ t} 
\end{gather}
where 
\begin{equation}
C_+=F_{ \zeroped{0} }'(V_{0}) \quad 
C_-=F_{ \zeroped{0} }'(V_{\infty})
\end{equation} 
and $\gamma=V_{\infty}- V_0$.

In absence of recollisions our model forecasts an exponential law approach to a limiting velocity that is what we expected from ordinary friction model, namely: 
$ \dot{V(t)} = E - b\,V $ , $ b > 0 $.

%
%
% Taken out
%
%
%

The study of the autonomous equation was trivial right because such a model neglects the interaction between the gas and the body. On the contrary with the full problem (including recollisions) we will have to deal with the internal coupling of our system.
The next section will be devoted to the study of the full problem through the two main theorems of this work.

%%%%%%%%%%%%%%%%%%%%%%%%%%%%%%%%%%%%%%%%%%%%%%  section 3   %%%%%%%%%%%%%%%%%%%%%%%%%%%%%%%%%%%%
%%%%%%%%%%%%%%%%%%%%%%%%%%%%%%%%%%%%%%%%%%%%%%  section 3   %%%%%%%%%%%%%%%%%%%%%%%%%%%%%%%%%%%%
%%%%%%%%%%%%%%%%%%%%%%%%%%%%%%%%%%%%%%%%%%%%%%  section 3   %%%%%%%%%%%%%%%%%%%%%%%%%%%%%%%%%%%%
%%%%%%%%%%%%%%%%%%%%%%%%%%%%%%%%%%%%%%%%%%%%%%  section 3   %%%%%%%%%%%%%%%%%%%%%%%%%%%%%%%%%%%%
%%%%%%%%%%%%%%%%%%%%%%%%%%%%%%%%%%%%%%%%%%%%%%  section 3   %%%%%%%%%%%%%%%%%%%%%%%%%%%%%%%%%%%%
%%%%%%%%%%%%%%%%%%%%%%%%%%%%%%%%%%%%%%%%%%%%%%  section 3   %%%%%%%%%%%%%%%%%%%%%%%%%%%%%%%%%%%%
%%%%%%%%%%%%%%%%%%%%%%%%%%%%%%%%%%%%%%%%%%%%%%  section 3   %%%%%%%%%%%%%%%%%%%%%%%%%%%%%%%%%%%%
%%%%%%%%%%%%%%%%%%%%%%%%%%%%%%%%%%%%%%%%%%%%%%  section 3   %%%%%%%%%%%%%%%%%%%%%%%%%%%%%%%%%%%%
%%%%%%%%%%%%%%%%%%%%%%%%%%%%%%%%%%%%%%%%%%%%%%  section 3   %%%%%%%%%%%%%%%%%%%%%%%%%%%%%%%%%%%%
%%%%%%%%%%%%%%%%%%%%%%%%%%%%%%%%%%%%%%%%%%%%%%  section 3   %%%%%%%%%%%%%%%%%%%%%%%%%%%%%%%%%%%%
%%%%%%%%%%%%%%%%%%%%%%%%%%%%%%%%%%%%%%%%%%%%%%  section 3   %%%%%%%%%%%%%%%%%%%%%%%%%%%%%%%%%%%%

\section{The full Problem}
\subsection{The theorems and solution strategy}
\label{SECTION The theorems}

We are now in position to state the main results of the present work.
\begin{theor}
There exists $\gamma _0= \gamma_0(\beta$, $\rho,E,R,h,V_{\infty})>0$ sufficiently small such that , for any $\gamma \in (0,\gamma _0)$
, there exists at least one solution $(V(t),f(t))$ to problem (\ref{vlasov})-(\ref{d/dt V = E- F(t)}). Moreover any solution $V(t)$ satisfies 
\begin{equation}
\gamma e^{-C_-t} \leq V_{\infty}- V(t) \leq \gamma e^{-C_+ t}+ 
\frac{A_+}{(1+t)^3}\gamma^3 \, ; \qquad \forall t \geq 0
 \end{equation}
for a suitable positive constant $A_+$ indipendent of $\gamma$.
\end{theor}

\begin{theor}
Let $\gamma \in (0,\gamma _0)$. There exists a sufficiently large $\overline{t}$, depending on $\gamma$, such that any solution $(V(t),f(t))$ to problem (\ref{vlasov})-(\ref{d/dt V = E- F(t)})  satisfies
\begin{equation}
V_{\infty}-V(t)\geq \gamma e ^{-C_-t} + \frac{A_- \gamma ^4}{t^3} \chi \{t>\overline{t} \}
\, ; \qquad \forall t \geq 0
\end{equation}
where  $A_-$ is a positive constant, indipendent of $\gamma$, and $ \chi \{... \} $ characteristic function of $ \{... \}. $

\end{theor}
We remind that both Theorems 3.1 and 3.2 hold in $d =2$ and $ d=3 $ dimensions.\\

We present here the strategy to prove Theorem $ 3.1 $.

If we consider an assigned velocity $ V_1 =W $ for the body we can compute the respective gas P.d.f. $ f_W $, or equivalently the recollision terms $ r_W^{\pm} $, by means of characteristics as explained in section (\ref{Section recollision terms}).
We then consider the modified problem
\begin{gather}
\label{d dt V_W = E -FW -rW}
\frac{d}{dt} V_W(t) = \frac{E-F_0(W(t))}{V_{\infty}-W(t)}(V_{\infty}-V_W(t)) -r_W^+(t)-r_W^-(t), \\
V_W(0) = V_0>0, 
\end{gather}
where
\begin{equation} 
\label{r+W}
 r_W^+(t) =2\rho (\frac{\beta}{\pi})^{\frac{3}{2}}\int\limits_{|\bm{x}_{\perp}|\leq R}
 d\bm{x}_{\perp} \int\limits_{v_x<W(t)} dv_x (v_x-W(t))^2\int \limits_{\mathbb{R}^2} d\bm{v}_{\perp} e^{-\beta \bm{v}_{\perp}^2} (e^{-\beta v_{0x}^2} -e^{- \beta v_x^2})  \\
\end{equation}
\begin{equation} 
\label{r-W}
r_W^-(t) =2\rho (\frac{\beta}{\pi})^{\frac{3}{2}}\int\limits_{|\bm{x}_{\perp}|\leq R}
d\bm{x}_{\perp} \int\limits_{v_x\geq W(t)} dv_x (v_x-W(t))^2\int \limits_{\mathbb{R}^2} d\bm{v}_{\perp} e^{-\beta \bm{v}_{\perp}^2} (e^{-\beta v_{x}^2} -e^{- \beta v_{0x}^2}).
\end{equation}
and $ v_{\zeroped{0}x} = v_x(0,t,X(t), \bm{x}_{\perp} ,v_x,\bm{v}_{\zeroped{\perp}}) $. Eq. (\ref{d dt V_W = E -FW -rW}) leads to  a new velocity $ V_2= V_W $.
We repeat the same argument with  $ V_2 $ as the new starting velocity and so on. We thus obtain a sequence  $ \{ V_n \, , n\geq 1 \} $ which we can formally express through the map :
\begin{gather}
V_{n+1}= \mathcal{F} (V_n) \notag \\
\label{map}
V_1(t)=W(t),
\end{gather}

Any solution of our problem is a fixed point of the map.
The proof of the theorem will be constructive: we will show how the map acts on a suitable class of functions and look for a set invariant under its action showing that functions belonging to this set enjoy properties quoted in Theorem $ 3.1 $.

The above mentioned suitable class of functions is represented by the set $\Omega_\alpha$ of $t$-a.e differentiable functions $t \to W(t) \in [V_0,V_\infty] $ such that:
%%%%%%%%%%%%%%%%%%%%%%%%%%%%%%%%%%%%%%%%%%%%%%%%%%%%%%%%%%%%%%%%%%%%%%%%%
\begin{subequations}
\label{Omega}
\begin{gather} 
\label{inequality a}
V_{\infty} -W(t)\geq \gamma e^{-C_- t}, \\
\label{inequality b}
 V_{\infty} -W(t)\leq \gamma e^{-C_+ t} +\gamma^3 \frac{A_+}{(1+t)^\alpha} \quad ; \quad 
 \alpha >1, \\
\frac{d}{dt} W(t)>0 \quad \forall t\in [0,t_0],
\end{gather}
\end{subequations}
with 
%%%%%%%%%%%%%%%%%%%%%%%%%%%%%%%%%%%%%%%%%%%%%%%%%%%%%%%%%%%%%%%%%%%%%%%%%
\begin{gather}
t_0 =\frac{1}{2C_-}Log[\frac{C_+}{\gamma}], \\
\gamma = V_\infty -V_0, \\
C_+ =F_0'(V_0)\leq C_- =F_0'(V_\infty), 
\end{gather}
%%%%%%%%%%%%%%%%%%%%%%%%%%%%%%%%%%%%%%%%%%%%%%%%%%%%%%%%%%%%%%%%%%%%%%%%%
and where $A_+$ is a positive constant independent of $\gamma$.
%This seemingly abstract choice has on the contrary a solid physics meaning, we address the interested reader to appendix for a detailed argumentation.

Before starting the proof we collect in the following Lemma some properties of the function 
$ W $, which will be useful in the sequel: for $ 0 \leq s \leq t $, we set 
%%%%%%%%%%%%%%%%%%%%%%%%%%%%%%%%%%%%%%%%%%%%%%%%%%%%%%%%%%%%%%%%%%%%%%%%%
\begin{gather}
\overline{W_{s,t}}=\frac{1}{t-s} \int_s^t W(\tau)d\tau
\end{gather}
%%%%%%%%%%%%%%%%%%%%%%%%%%%%%%%%%%%%%%%%%%%%%%%%%%%%%%%%%%%%%%%%%%%%%%%%%
and
%%%%%%%%%%%%%%%%%%%%%%%%%%%%%%%%%%%%%%%%%%%%%%%%%%%%%%%%%%%%%%%%%%%%%%%%%
\begin{gather}
\overline{W_{t}}=\overline{W_{0,t}}\, . 
\end{gather}
%%%%%%%%%%%%%%%%%%%%%%%%%%%%%%%%%%%%%%%%%%%%%%%%%%%%%%%%%%%%%%%%%%%%%%%%%
In what follows the symbol $C$ will indicate any positive constant, independent from $\gamma$ which is our small parameter. Any such constant is explicitly computable.

\begin{Lems}
\label{lemmi W_s,t}
Let $W(t) \in \Omega_\alpha \ ; \ \alpha>1 $. Suppose $\gamma$ sufficiently small. \\
Then $\forall t>0$:
%%%%%%%%%%%%%%%%%%%%%%%%%%%%%%%%%%%%%%%%%%%%%%%%%%%%%%%%%%%%%%%%%%%%%%%%%
\begin{subequations}
\label{Lemmas properties}
\begin{gather} 
W(t)>\overline{W_{t}} \ , \quad  \\
\frac{d}{dt} \overline{W_{t}}  > 0 ,  \\
\label{W_s,t  > W_t}
\overline{W_{s,t}}>\overline{W_{t}} \quad ; \quad \forall s \in (0,t), \\
W(t)-\overline{W_{t}} \leq \frac{C}{(1+t)}(\gamma+ A_+ \gamma^3).
\end{gather}
\end{subequations}
\end{Lems}

\begin{proof}
a) The result is trivially true for $t<t_0$ because in this region $W$ is increasing.
For $t>t_0$ 
\begin{align}
W(t)-\overline{W_t} & =\frac{1}{t}\int _0 ^t d \tau [W(t)-W(\tau)] \notag \\
                    & =\frac{1}{t}\int _0 ^t d \tau [W(t)-V_{\infty}+V_{\infty}-W(\tau)] \notag \\
                    & \geq -\gamma e^{-C_+t} -\frac{A_+ \gamma ^3}{(1+t)^{\alpha}}+\frac{1}{t} \int _0^t d \tau \ \gamma e^{-C_-\tau} \notag \\
                    &= -\gamma e^{-C_+t} -\frac{A_+ \gamma ^3}{(1+t)^{\alpha}}+\frac{\gamma}{C_-t}-\frac{\gamma}{C_-t} e^{-C_-t} \notag
\end{align}
The last quantity is positive by taking $\gamma$ sufficiently small and consequently $t_0$ sufficiently large; in this way $t$ is large and the term $\frac{\gamma}{C_-t}$ is dominant because $ \alpha > 1 $.\\
b)\begin{equation}
\frac{d}{dt}\overline{W_t}=-\frac{1}{t^2}\int_0^tW(\tau)d \tau +\frac{1}{t}W(t)=\frac{-\overline{W_t}+W(t)}{t}>0 \notag
\end{equation} \\
the last inequality holding by virtue of (a).

c) \begin{align}\overline{W_{s,t}}-\overline{W_t} &=\frac{1}{t-s}\int _s^t W(\tau) d \tau-\frac{1}{t}\int _0^t W(\tau) d \tau \notag \\
                                  &=\frac{1}{t-s} \Bigl (\int _0^t W(\tau) d \tau)-\int_0^s W(\tau)d\tau \Bigr) -\frac{1}{t}\int _0^t W(\tau) d \tau \notag \\
                                  &=\frac{s}{t(t-s)}\int _0^t W(\tau)d\tau-\frac{1}{t-s} \int_ 0^s W(\tau)d\tau \notag \\
                                  &=\frac{s}{t-s} (\overline{W_t}-\overline{W_s}) \notag
\end{align}
which is positive by (b). \\
d)
\begin{align}
W(t)-\overline{W_t} &= \frac{1}{t} \int _0^t d\tau [W(t)-W(\tau)] \notag \\
                     &\leq  \frac{1}{t} \int _0^t d\tau [V_{\infty}-W(\tau)] \notag \\
                     & \leq \frac{1}{t} \int _0^t d\tau [ \gamma e^{-C_+ \tau} + \frac{A \gamma ^3}{(1+\tau)^{\alpha}}] \notag
\end{align}
We now observe that both the functions $u(t)=\frac{1}{t}\int_0^t d\tau e^{-C_+ \tau}$ and $z(t)=\frac{1}{t}\int_0^t d\tau \frac{1}{(1+\tau)^{\alpha}}$ are bounded and decay as $\frac{1}{1+t} $,    we can then conclude that $u(t)\leq \frac{C_1}{1+t}$ and 
$z(t)\leq \frac{C_2}{1+t}$ and finally
\begin{equation}
W(t)-\overline{W_t} \leq \frac{C}{1+t}(\gamma +A \gamma ^3) \notag
\end{equation}
\end{proof}

%%%%%%%%%%%%%%%%%%%%%%%%%%%%%%%%%%%%%%%%%%%%%%%%%%%%%%%%%%%%%%%%%%%%%%%%%
%%%%%%%%%%%%%%%%%%%%%%%%%%%%%%%%%%%%%%%%%%%%%%%%%%%%%%%%%%%%%%%%%%%%%%%%%
%%%%%%%%%%%%%%%%%%%%%%%%%%%%%%%%%%%%%%%%%%%%%%%%%%%%%%%%%%%%%%%%%%%%%%%%%
%%%%%%%%%%%%%%%%%%%%%%%%%%%%%%%%%%%%%%%%%%%%%%%%%%%%%%%%%%%%%%%%%%%%%%%%%
%%%%%%%%%%%%%%%%%%%%%%%%%%%%%%%%%%%%%%%%%%%%%%%%%%%%%%%%%%%%%%%%%%%%%%%%%
%%%%%%%%%%%%%%%%%%%%%%%%%%%%%%%%%%%%%%%%%%%%%%%%%%%%%%%%%%%%%%%%%%%%%%%%%
%%%%%%%%%%%%%%%%%%%%%%%%%%%%%%%%%%%%%%%%%%%%%%%%%%%%%%%%%%%%%%%%%%%%%%%%%

\subsection{Computation of Recollision terms}
\label{SECTION Computation Recollision terms}
As explained in the previous section, the proof starts studying how the map acts on the set
$ \Omega_{\alpha} $, in order to do this we preliminarily have to estimate $ r_W^{\pm}(t) $.

The computation of $r_W^+$ concerns frontal recollisions. 
The position-velocity of the particle on which we integrate in $r_W^+$ is:
$ (x,\bm{x}_{\zeroped{\perp}},v_x,\bm{v}_{\zeroped{\perp}} ) $ with 
$ x=X(t) $ , $ | \bm{x}_{\zeroped{\perp}} | \leq R $ , $ v_x \leq W(t) $.

%Here lateral barriers are present of length $h$ that can trap particles otherwise free to 
%escape from the disk.

The region of $ (\bm{x} , \bm{v}) $ that doesn't lead to recollision implies 
$ \bm{v}_{\zeroped{0}x} = \bm{v}_x  $ and thus has zero contribution on $r_W^+$.

We then compute $r_W^+$ by tracing back the time evolution of gas particles which showed recollision in the past. More precisely: if a collision occurs at time $t$, to have a past recollision a time 
$s\in (0,t)$ has to exist such that 
%%%%%%%%%%%%%%%%%%%%%%%%%%%%%%%%%%%%%%%%%%%%%%%%%%%%%%%%%%%%%%%%%%%%%%%%%
\begin{gather}
\label{vx = W_s,t}
v_x(t-s) =\int_s^t W(\tau) d\tau,
\end{gather}
%%%%%%%%%%%%%%%%%%%%%%%%%%%%%%%%%%%%%%%%%%%%%%%%%%%%%%%%%%%%%%%%%%%%%%%%%
that is $v_x = \overline{W_{s,t}}$ for some $s \in (0,t)$. \\
Condition (\ref{vx = W_s,t}) holds because every collision the particle has against lateral barriers doesn't change the momentum along the x-axis, that is $v_x$.

Besides by (\ref{W_s,t  > W_t}) in Lemma \ref{lemmi W_s,t}, the above condition gets:
\begin{gather}
\label{v_x > W_t}
v_x  \geq \overline{W_{t}}.
\end{gather}

The presence of the barriers makes the recollision condition on $\bm{v}_{\zeroped{\perp}}$ more delicate.\\
Let $x(\tau)$ be the position (on the x-axis) at time $\tau$ of a particle with velocity $v_x$, colliding frontally with the disk at time t, and which had a previous collision at time $s$,  i.e.
%%%%%%%%%%%%%%%%%%%%%%%%%%%%%%%%%%%%%%%%%%%%%%%%%%%%%%%%%%%%%%%%%%%%%%%%%
\begin{gather}
\label{x_s}
x(\tau)= X(s)+v_x(\tau -s)\quad ; \quad \tau \in [s,t],
\end{gather}
where $v_x$ satisfies the recollision equation (\ref{vx = W_s,t}) so that
%%%%%%%%%%%%%%%%%%%%%%%%%%%%%%%%%%%%%%%%%%%%%%%%%%%%%%%%%%%%%%%%%%%%%%%%%
\begin{subequations}
\begin{gather}
x(s)=X(s), \\
x(t)=x=X(t) .
\end{gather}
\end{subequations}

%%%%%%%%%%%%%%%%%%%%%%%%%%%%%%%%%%%%%%%%%%%%%%%%%%%%%%%%%%%%%%%%%%%%%%%%%
Now,  for a given $s(v_x)$ the particle can either stay within the barriers for the whole time $ t-s $, 
or it can move outside during one or multiple intervals of times $\delta_i $ ;   we define  $ \Delta $ as the largest of these possible intervals, namely $ \Delta= \underset{i } {Max} \ ( \delta_i ) $. Obviously if the particle stays inside the body for the whole time, there are no such intervals $\delta_i $ and $ \Delta=0 $.

If $ n $ is the number of times the particle is outside during an interval $\delta_i >0 $ we define
$ \tau_i^{(a)} < \tau_i^{(b)} , \quad  i=1,2 \dots n  \quad $  as those times such that
%%%%%%%%%%%%%%%%%%%%%%%%%%%%%%%%%%%%%%%%%%%%%%%%%%%%%%%%%%%%%%%%%%%%%%%%%
%%%%%%%%%%%%%%%%%%%%%%%%%%%%%%%%%%%%%%%%%%%%%%%%%%%%%%%%%%%%%%%%%%%%%%%%%
\begin{gather}
\label{tau solution}
x(\tau_i^{(\eta)})-X(\tau_i^{(\eta)})=h \quad ; \quad \eta = a,b \ ,\\
x(\tau)-X(\tau)>h \quad \forall \tau \in (\tau_i^{(a)}, \tau_i^{(b)}), \\
\tau_{i+1}^{(a)} < \tau_{i+1}^{(b)}<\tau_{i}^{(a)}<\tau_{i}^{(b)} \quad  i=1,2 \dots n -1 
\end{gather}
%%%%%%%%%%%%%%%%%%%%%%%%%%%%%%%%%%%%%%%%%%%%%%%%%%%%%%%%%%%%%%%%%%%%%%%%%
that is to say that $  \tau_{i}^{(a)}$  , in the natural forward evolution, are times when the particle escapes the cylinder, while  
$ \tau_{i}^{(b)}$ are times of entering back the cylinder.
As clear from equation (\ref{x_s}) and  (\ref{tau solution}), $\tau_i$ depends only on $t$, $s$ ($v_x$ being determined from (\ref{vx = W_s,t}) ) and on $h$ as a parameter of our problem, 
we can then  write the $ n $ intervals of time during which the particle is outside the barriers as
\begin{gather}
\delta_i(s,t)=\tau_{i}^{(b)} - \tau_{i}^{(a)}  \quad  i=1,2 \dots n
\end{gather}
%%%%%%%%%%%%%%%%%%%%%%%%%%%%%%%%%%%%%%%%%%%%%%%%%%%%%%%%%%%%%%%%%%%%%%%%%
Let now $\bm{x}_{\zeroped{\perp}}(\tau) $ and $\bm{v}_{\zeroped{\perp}}(\tau) $ be respectively the position and the velocity of the particle along the $\bm{x}_{\zeroped{\perp}}$-axis at time $\tau$,  so that $ |\bm{x}_{\zeroped{\perp}}(s) | < R $ and $  \bm{x}_{\zeroped{\perp}}(t) = \bm{x}_{\zeroped{\perp}}  $.

It is only during intervals $\delta_i$ that particles are free to move in region 
$ |\bm{x}_{\zeroped{\perp}}(\tau) | > R$ , therefore in order for the particle to have a recollision at time $s$ the condition  on $\bm{v}_{\zeroped{\perp}}$ is the following 
%%%%%%%%%%%%%%%%%%%%%%%%%%%%%%%%%%%%%%%%%%%%%%%%%%%%%%%%%%%%%%%%%%%%%%%%%
\begin{gather}
\label{recollision v_y prima}
\lvert \bm{x}_{\zeroped{\perp}}(\tau_i^{(b)})-\bm{v}_{\zeroped{\perp}}(\tau_i^{(b)})(\tau_i^{(b)}-\tau_i^{(a)})\rvert < R 
\quad ; \quad  \forall i 
\end{gather}
%%%%%%%%%%%%%%%%%%%%%%%%%%%%%%%%%%%%%%%%%%%%%%%%%%%%%%%%%%%%%%%%%%%%%%%%%
A necessary condition for each equation in (\ref{recollision v_y prima}) to hold is
\begin{gather}
| \bm{v}_{\zeroped{\perp}}(\tau_i^{(b)})|\, \delta_i(s,t) = 
| \bm{v}_{\zeroped{\perp}}|\, \delta_i(s,t)  \leq  2R \quad ; 
\quad  \forall i 
\end{gather}
where we used the conservation of the transverse velocity $ \bm{v}_{\zeroped{\perp}} .$\\
The above condition is  equivalent to
\begin{gather}
| \bm{v}_{\zeroped{\perp}} |\, \Delta(s,t)  \leq  2R 
\end{gather}
where
\begin{gather}
\label{DELTA definition}
\Delta(s,t)= \underset{i \in [1,n]}{Max} \ \delta_i(s,t),
\end{gather}
%%%%%%%%%%%%%%%%%%%%%%%%%%%%%%%%%%%%%%%%%%%%%%%%%%%%%%%%%%%%%%%%%%%%%%%%%

Finally for a recollision to happen it is necessary that
\begin{gather}
\label{recollision conditions}
v_x \geq \overline{W_{t}} \qquad \text{ and } \qquad  \lvert \bm{v}_{\zeroped{\perp}} \rvert \, 
\Delta(s(v_x),t) \leq 2R . 
\end{gather}

Notice that for $v_x$ such that $\Delta(s(v_x),t)=0$, condition (\ref{recollision conditions}) implies\\
$\bm{v}_{\zeroped{\perp}} \in \mathbb{R}^2$ (or, in two dimensions, $ v_y \in \mathbb{R}^1 $), in other words if the particle never exits the cylinder no matter what the transverse velocity is, the particle is going to collide back $ D(t) $ (obviously provided $ v_x \geq \overline{W_{t}}  $). 

We now observe that the recollision conditions studied in \cite{nostro} (disk accelerated in a free Vlasov gas) can be deducted as a special case of the cylinder when $ h=0 $.
Indeed in this case equation (\ref{tau solution}) becomes 
%%%%%%%%%%%%%%%%%%%%%%%%%%%%%%%%%%%%%%%%%%%%%%%%%%%%%%%%%%%%%%%%%%%%%%%%%
\begin{gather}
x(\tau)-X(\tau)=0 \quad ; \quad \tau \in [s,t] , 
\end{gather}
%%%%%%%%%%%%%%%%%%%%%%%%%%%%%%%%%%%%%%%%%%%%%%%%%%%%%%%%%%%%%%%%%%%%%%%%%
with two solutions $\tau^{(a)}=s$ , $\tau^{(b)}=t $ and
%%%%%%%%%%%%%%%%%%%%%%%%%%%%%%%%%%%%%%%%%%%%%%%%%%%%%%%%%%%%%%%%%%%%%%%%%
\begin{gather}
\Delta(s,t)=t-s(v_x) ,
\end{gather}
as expected because $ h=0 $ is the simple case in which the particle is compelled to stay ''out''  during the whole interval $(s,t)$.
This gives back the necessary condition of recollision on $\bm{v}_{\zeroped{\perp}}$ for the simple disk in \cite{nostro}, namely:
\begin{gather}
v_x \geq \overline{W_{t}} \qquad \text{ and } \qquad \lvert \bm{v}_{\zeroped{\perp}} \rvert \leq \frac{2R}{t-s(v_x)} .
\end{gather}

%%%%%%%%%%%%%%%%%%%%%%%%%%%%%%%%%%%%%%%%%%%%%%%%%%%%%%%%%%%%%%%%%%%%%%%%%
%%%%%%%%%%%%%%%%%%%%%%%%%%%%%%%%%%%%%%%%%%%%%%%%%%%%%%%%%%%%%%%%%%%%%%%%%
%%%%%%%%%%%%%%%%%%%%%%%%%%%%%%%%%%%%%%%%%%%%%%%%%%%%%%%%%%%%%%%%%%%%%%%%%
%%%%%%%%%%%%%%%%%%%%%%%%%%%%%%%%%%%%%%%%%%%%%%%%%%%%%%%%%%%%%%%%%%%%%%%%%
%%%%%%%%%%%%%%%%%%%%%%%%%%%%%%%%%%%%%%%%%%%%%%%%%%%%%%%%%%%%%%%%%%%%%%%%%
%%%%%%%%%%%%%%%%%%%%%%%%%%%%%%%%%%%%%%%%%%%%%%%%%%%%%%%%%%%%%%%%%%%%%%%%%
%%%%%%%%%%%%%%%%%%%%%%%%%%%%%%%%%%%%%%%%%%%%%%%%%%%%%%%%%%%%%%%%%%%%%%%%%

\subsection{Escaping times}
\label{SECTION kinematic issues}
As pointed out in the previous section, the region of recollision for $\bm{v}_{\zeroped{\perp}}$ strongly depends on the behaviour of the particle during the interval $[s,t]$. In particular, we need a criterion to determine whether the particle moves outside the barriers or remains constantly in the cylinder.
To this aim we introduce the following function which represents the distance between the particle and the disk, namely:
\begin{gather}
d(s,t,\tau)= \int _s ^{\tau} dp \ ( \overline{W_{s,t}}-W(p)) \qquad \tau \in [s,t],
\end{gather}

Before stating an important result, we define two auxiliary functions (Figure 2): 
\begin{gather}
\label{defi-g} g(s,\tau)= \lim_{t\to \infty} d(s,t,\tau)= \int_s ^{\tau} dp \ (V_{\infty}-W(p)), \\
\label{def-f} f(s)= \lim_{\tau \to \infty} g(s,\tau)=\int_s ^{\infty} dp \ (V_{\infty}-W(p)).
\end{gather}                                                                                                                                                
given $ 0 \leq s \leq \tau \leq t $.
Inequality (\ref{inequality b}) ensures both the correctness of the limit (\ref{defi-g}) (in particular the fact that $ \lim_{t \to \infty} V_{\infty}- \overline{W_{st}} =0$), and the convergence of the integral (\ref{def-f}).

\begin{figure}

\includegraphics[scale=0.28]{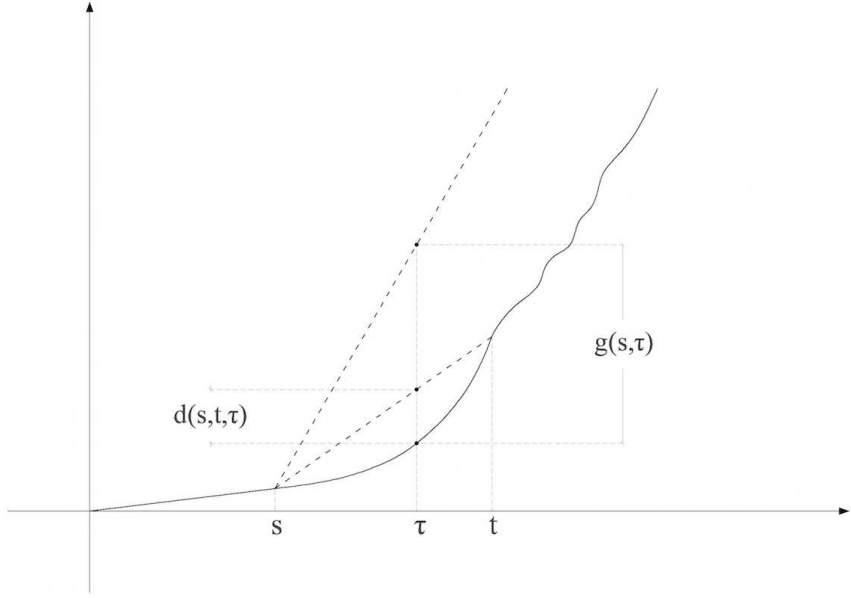}
\caption{  Example of a trajectory of the cylinder base (thick line) and two particle trajectories (dashed lines), both particles collide with the cylinder base at time $  s $; while the first particle collides again at the finite time $  t $, the second collides again after an infinite time. The horizontal axis represents the time, the vertical axis represents the position along the $  x $ axis. }
\end{figure}

We preliminarily notice that the function $ f(s) $ is strictly decreasing 
and bounded from above:
\begin{gather}
\label{df/ds} 
\frac{df}{ds}(s)=W(s)-V_{\infty}<0, \\
\label{f(s)<C)}
f(s)<f(0)=\int_0 ^{\infty} dp (V_{\infty}-W(p))<C.
\end{gather}
moreover
\begin{gather}
\label{d(s,t,)< f(s)}
d(s,t,\tau) <  \int _s ^{\tau} dp \ ( V_{\infty} - W(p) ) \leq f(s).
\end{gather}

Let now $s^*$ be the possible solution of 
\begin{gather}
\label{f(s*= h}
f(s^*) = h ,
\end{gather}
where we remind the reader that $ h $ is the length of the cylinder side; properties (\ref{df/ds}) and (\ref{f(s)<C)}) ensure that there is at most one $s^* > 0 $ solution to (\ref{f(s*= h}).\\
We summarize an important result in the following lemma.

\begin{Lems}
\label{Proposition d(s,t,tau)}
For $ h\, $ such that $  $ $ \exists s^* > 0$ solution of (\ref{f(s*= h}) then $ \Delta(s(v_x),t)=0   $ \\ $ \forall s^* < s < t $. \\
For $ h\, $ such that $  $ $ \nexists s^* > 0$ solution of (\ref{f(s*= h}) then $ \Delta(s(v_x),t)=0 \quad \forall \, 0< s < t $ .

\end{Lems}
\begin{proof} 

If $ h \geq f(0)$, by properties (\ref{df/ds}) and (\ref{f(s)<C)}) $ \nexists s^* > 0$ solution of (\ref{f(s*= h}), in this case 
$ f(s) < f(0) < h $ $ \ \forall \,s > 0 $ and together with eq.(\ref{d(s,t,)< f(s)})\\
we have 
\begin{gather*}
d(s,t,\tau) <h \quad \forall \, 0 < s < t ;
\end{gather*} 
that means the particle is trapped in the cylinder during the whole interval $ [s,t] $, thus
$ \Delta(s,t) = 0  \quad \forall \,  0 < s < t $.\\
\

On the other hand if $ h < f(0) $ the same properties of above guarantee the existence of 
$ s^* > 0$, in this case $ f(s) < h  $ $ \ \forall \, s^* < s < t  $ 
 and together with eq.(\ref{d(s,t,)< f(s)}) we have
\begin{gather*}
d(s,t,\tau) < h \quad \forall  s^* < s < t
\end{gather*}
that is
$ \Delta(s,t) = 0  \quad \forall \,  s^* < s < t $.

\end{proof}

We now prove that $s^*$ is bounded from above and from below in the following way:
\begin{align}
\label{s-< s* < s+}
s_- (\gamma) \leq s^* \leq s_+ (\gamma).
\end{align}
where
\begin{gather}
 s_- (\gamma)=\frac{1}{C_-} log \frac{\gamma}{hC_-}, \\
 s_+(\gamma) = \sqrt[\alpha -1]{\frac{C \gamma}{h}}-1.  
\end{gather}
Indeed,
\begin{align}
f(s^*)&=\int_{s^*} ^ \infty dp \ (V_{\infty}-W(p)) \notag \\
&\leq \int_{s^*} ^ \infty dp \ ( \gamma e^{-C_+ p} + \frac{A \gamma ^3}{(1+p)^{\alpha}}) \notag  \\
&=\frac{\gamma}{C_+} e^{-C_+ s^*}+ \frac{A \gamma ^3}{\alpha -1} \frac{1}{(1+s^*)^{\alpha -1}}.
\end{align}

We observe that 
$e^{-C_+ s^*} \leq \frac{C_1}{(1+s^*)^{\alpha -1}}$ for a suitable constant $C_1$ , while $A \gamma ^3 \leq C_2 \gamma $ for  $\gamma $ sufficiently small. We therefore obtain 
\begin{gather}
f(s^*) \leq \frac{C \gamma}{(1+s^*)^{\alpha -1}} ,
\end{gather}
which, together with definition (\ref{f(s*= h}), yields 
\begin{gather}
s^* \leq \sqrt[\alpha -1]{\frac{C \gamma}{h}}-1
\end{gather}
On the other hand,
\begin{align}
f(s^*)&=\int_{s^*} ^ \infty dp \ (V_{\infty}-W(p)) \notag \\
& \geq \int_{s^*} ^ \infty dp \ \gamma  e^{-C_- p} \notag \\
& = \frac{\gamma}{C_-} e^{-C_- s^*}
\end{align}
which, together with definition (\ref{f(s*= h}) , gives
\begin{gather}
s^* \geq \frac{1}{C_-} log \frac{\gamma}{hC_-} 
\end{gather}

It is now clear from (\ref{s-< s* < s+}) that the magnitude of $s^*$ increases with $\gamma$ ; in particular for $\gamma$ small enough  there is no positive $s^*$  satisfying (\ref{f(s*= h}).

The quantity  $ \Theta_{\gamma} = \frac{h} {\gamma}  $ is a natural parameter of our problem which gives an indication on when the presence of the barriers becomes relevant in the dynamic. Indeed,
in the limiting case $  \Theta_{\gamma}  \rightarrow 0 $, condition (\ref{s-< s* < s+})   ensures that
$ s^* \rightarrow \infty $ and the trapping effect becomes  clearly irrelevant.

Now, as it will be clear in the next section, the dominant contribution that will be responsible for the $ t^{-3} $ trend, comes essentially from the region
$ s^* < s < t  $ but this is  achievable only for larger and larger $ t $ as 
$  \Theta_{\gamma}  $ decrease.

%%%%%%%%%%%%%%%%%%%%%%%%%%%%%%%%%%%%%%%%%%%%%%%%%%%%%%%%%%%%%%%%%%%%%%%%%%%%%%%
%%%%%%%%%%%%%%%%%%%%%%%%%%%%%%%%%%%%%%%%%%%%%%%%%%%%%%%%%%%%%%%%%%%%%%%%%%%%%%%%%%
%%%%%%%%%%%%%%%%%%%%%%%%%%%%%%%%%%%%%%%%%%%%%%%%%%%%%%%%%%%%%%%%%%%%%%%%%%%%
%%%%%%%%%%%%%%%%%%%%%%%%%%%%%%%%%%%%%%%%%%%%%%%%%%%%%%%%%%%%%%%%%%%%%%%%%%%%%%%
%%%%%%%%%%%%%%%%%%%%%%%%%%%%%%%%%%%%%%%%%%%%%%%%%%%%%%%%%%%%%%%%%%%%%%%%%%%%%%%%%%
%%%%%%%%%%%%%%%%%%%%%%%%%%%%%%%%%%%%%%%%%%%%%%%%%%%%%%%%%%%%%%%%%%%%%%%%%%%%
%%%%%%%%%%%%%%%%%%%%%%%%%%%%%%%%%%%%%%%%%%%%%%%%%%%%%%%%%%%%%%%%%%%%%%%%%%%%%%%
%%%%%%%%%%%%%%%%%%%%%%%%%%%%%%%%%%%%%%%%%%%%%%%%%%%%%%%%%%%%%%%%%%%%%%%%%%%%%%%%%%
%%%%%%%%%%%%%%%%%%%%%%%%%%%%%%%%%%%%%%%%%%%%%%%%%%%%%%%%%%%%%%%%%%%%%%%%%%%%

\subsection{Estimate of $r_W ^+$ and $r_W ^-$  }
\label{SECTION estimate r+}

We are now in position to show the behavior of recollision terms, that is summarized in the next lemma. 
\begin{Lems}
Let $W \in \Omega_\alpha \ , \  \alpha >1 $. Then $\forall t \geq 0$,
\begin{gather}
\label{r_+ behaviour lemma}
r_W^+(t) \leq C \frac{(\gamma+A_+ \gamma^3)^3}{ (1+t)^3 } .\\
r_W^- \leq C \chi \{t>t_0 \} \bigl(\frac{\gamma+A_+ \gamma^3}{(1+t)^\alpha} \bigl)^3.
\end{gather}
\end{Lems}
%%%%%%%%%%%%%%%%%%%%%%%%%%%%%%%%%%%%%%%%%%%%%%%%%%%%%%%%%%%%%%%%%%%%%%%%%
\begin{proof}
%%%%%%%%%%%%%%%%%%%%%%%%%%%%%%%%%%%%%%%%%%%%%%%%%%%%%%%%%%%%%%%%%%%%%%%%%ox
%%%%%%%%%%%%%%%%%%%%%%%%%%%%%%%%%%%%%%%%%%%%%%%%%%%%%%%%%%%%%%%%%%%%%%%%%
From recollision condition (\ref{recollision conditions}) and using that 
$e^{-\beta v_{0x}^2} - e^{-\beta v_{x}^2} \leq 1 $, in expression (\ref{r+W}), it follows:
%%%%%%%%%%%%%%%%%%%%%%%%%%%%%%%%%%%%%%%%%%%%%%%%%%%%%%%%%%%%%%%%%%%%%%%%%
\begin{gather}
\label{r+ v_y only}
r_W^+(t) \leq C \int \limits_{\overline{W_t} }^{W(t)} dv_x \ (W(t)-v_x)^2 
				\int \limits_{\mathbb{R}^2 } d\bm{v}_{\zeroped{\perp}} \ e^{-\beta \bm{v}_{\zeroped{\perp}}^2} \,
				\chi ( \{\, \lvert \bm{v}_{\zeroped{\perp}} \rvert \, \Delta(s,t) \leq 2R \, \} ).
\end{gather}
%%%%%%%%%%%%%%%%%%%%%%%%%%%%%%%%%%%%%%%%%%%%%%%%%%%%%%%%%%%%%%%%%%%%%%%%%
As we said, if $\gamma$ is too small, no positive $s^*>0$ solves eq. (\ref{f(s*= h}).\\
In this case $\Delta(s,t)=0 \quad \forall s > 0 $, consequently
%%%%%%%%%%%%%%%%%%%%%%%%%%%%%%%%%%%%%%%%%%%%%%%%%%%%%%%%%%%%%%%%%%%%%%%%%
\begin{align}
\label{int e^-bv_y}
\int \limits_{\mathbb{R}^2 } d\bm{v}_{\zeroped{\perp}} \ e^{-\beta \bm{v}_{\zeroped{\perp}}^2} \,
		\chi ( \{\, \lvert \bm{v}_{\zeroped{\perp}} \rvert \, \Delta(s,t) \leq 2R \, \} )
		= \int \limits_{\mathbb{R}^2 } d\bm{v}_{\zeroped{\perp}} \ e^{-\beta \bm{v}_{\zeroped{\perp}}^2}
		\leq C
\end{align}

and from properties of $ W(t) $ in Lemma \ref{lemmi W_s,t} we immediately obtain
\begin{gather}
\label{r+ v_y only}
r_W^+(t) \leq C \int \limits_{\overline{W_t} }^{W(t)} dv_x \ (W(t)-v_x)^2 \leq C \frac{(\gamma+A_+ \gamma^3)^3}{ (1+t)^3 } .
\end{gather}
%%%%%%%%%%%%%%%%%%%%%%%%%%%%%%%%%%%%%%%%%%%%%%%%%%%%%%%%%%%%%%%%%%%%%%%%%
Instead, if $s^*>0$ exists, a finer analysis is needed.
First we study (\ref{r+}) for large $t$, in particular $ s^* \leq \frac{t}{2}$, this can be  done since $s^*$ was proven to be bounded in $t$.\\
We will analyse the contribution to $r_W^+$ coming from different region of $v_x$.\\
The first region is $ \frac{t}{2} \leq s(v_x) $, where $\Delta(s,t)=0 $ and
(\ref{int e^-bv_y} ) still holds.
%%%%%%%%%%%%%%%%%%%%%%%%%%%%%%%%%%%%%%%%%%%%%%%%%%%%%%%%%%%%%%%%%%%%%%%%%

%%%%%%%%%%%%%%%%%%%%%%%%%%%%%%%%%%%%%%%%%%%%%%%%%%%%%%%%%%%%%%%%%%%%%%%%%

%%%%%%%%%%%%%%%%%%%%%%%%%%%%%%%%%%%%%%%%%%%%%%%%%%%%%%%%%%%%%%%%%%%%%%%%%
Moreover in this region from (\ref{vx = W_s,t}) and (\ref{Omega}) it yields: 
\begin{align}
W(t)-v_x &=\frac{1}{t-s} \int_s^t ( W(t)-W(\tau) ) d\tau \notag \\
		 &\leq\frac{1}{t-s} \int_s^t ( V_\infty - W(\tau) ) d\tau \notag \\
		 &\leq \frac{1}{t-s} \int_s^t \frac{ A_+ \gamma^3}{(1+\tau)^\alpha} +\gamma e^{-C_+ \tau} d\tau
		 \notag \\
		 &\leq \gamma e^{-C_+ s} + \frac{ A_+ \gamma^3}{(1+s)^\alpha}.
\end{align}
%%%%%%%%%%%%%%%%%%%%%%%%%%%%%%%%%%%%%%%%%%%%%%%%%%%%%%%%%%%%%%%%%%%%%%%%%
Since
%%%%%%%%%%%%%%%%%%%%%%%%%%%%%%%%%%%%%%%%%%%%%%%%%%%%%%%%%%%%%%%%%%%%%%%%%
\begin{gather*}
e^{-C_+ s} \leq C \frac{1}{(1+s)^\alpha} \quad \text{for some} \, C>0 \quad \text{and} \quad
\frac{1}{(1+s)^\alpha} < \frac{2^\alpha}{(1+t)^\alpha}, 
\end{gather*}
%%%%%%%%%%%%%%%%%%%%%%%%%%%%%%%%%%%%%%%%%%%%%%%%%%%%%%%%%%%%%%%%%%%%%%%%%
it follows that
%%%%%%%%%%%%%%%%%%%%%%%%%%%%%%%%%%%%%%%%%%%%%%%%%%%%%%%%%%%%%%%%%%%%%%%%%
\begin{gather}
W(t)-v_x \leq C \frac{A_+ \gamma^3 +\gamma}{(1+t)^\alpha}.
\end{gather}
%%%%%%%%%%%%%%%%%%%%%%%%%%%%%%%%%%%%%%%%%%%%%%%%%%%%%%%%%%%%%%%%%%%%%%%%%
Therefore the first contribution to $r_W^+$:
\begin{align}
\label{1st contribution}
&	C \int\limits_{\frac{t}{2} \leq s(v_x) \leq t} dv_x \ (W(t)-v_x)^2   \notag \\
&	\leq C \int_{0}^\infty dv_x \ (W(t)-v_x)^2  \, 
	  \chi \{ \, W(t)-\frac{A_+ \gamma^3 +\gamma}{(1+t)^\alpha} \leq v_x \leq W(t) \, \}  \notag \\
&	\leq C \Bigl [\frac{A_+ \gamma^3 +\gamma}{(1+t)^\alpha} \Bigl ]^3.
\end{align}
%%%%%%%%%%%%%%%%%%%%%%%%%%%%%%%%%%%%%%%%%%%%%%%%%%%%%%%%%%%%%%%%%%%%%%%%%
Similarly for $v_x$ such that $ s^* \leq  s(v_x)\leq \frac{t}{2} $ relation (\ref{int e^-bv_y}) still holds.
On the other hand in this region  it can't be found for $(W(t)-v_x)$ an upper bound finer than
\begin{align}
W(t)-v_x = W(t) -\overline{W_{s,t}} \leq W(t) -\overline{W_{t}} \leq
	 \frac{C}{(1+t)}(\gamma+ A_+ \gamma^3)
\end{align}
%%%%%%%%%%%%%%%%%%%%%%%%%%%%%%%%%%%%%%%%%%%%%%%%%%%%%%%%%%%%%%%%%%%%%%%%%
which holds in general by (\ref{Lemmas properties}). For this reason the second contribution reads 
%%%%%%%%%%%%%%%%%%%%%%%%%%%%%%%%%%%%%%%%%%%%%%%%%%%%%%%%%%%%%%%%%%%%%%%%%
\begin{align}
\label{2nd contribution}
&	C \int\limits_{ s^* \leq  s(v_x)\leq \frac{t}{2} } dv_x \ (W(t)-v_x)^2 
	\leq C \int\limits_{ \overline{W_{t}} }^{W(t)} dv_x \ (W(t)-v_x)^2  \notag \\
&	\leq C \, (W(t)-\overline{W_t})^3 \leq C \Bigl [\frac{A_+ \gamma^3 +\gamma}{(1+t)} \Bigl ]^3.
\end{align}
%%%%%%%%%%%%%%%%%%%%%%%%%%%%%%%%%%%%%%%%%%%%%%%%%%%%%%%%%%%%%%%%%%%%%%%%%
As last contribution, $s(v_x) \leq s^*$, we have
%%%%%%%%%%%%%%%%%%%%%%%%%%%%%%%%%%%%%%%%%%%%%%%%%%%%%%%%%%%%%%%%%%%%%%%%%
\begin{align}
\label{3rd contribution first expression}
&	C \int \limits_{s(v_x) \leq s^*} dv_x \ (W(t)-v_x)^2 
				\int \limits_{\mathbb{R}^2 } d\bm{v}_{\zeroped{\perp}} \ e^{-\beta \bm{v}_{\zeroped{\perp}}^2} \,
				\chi ( \{\, \lvert \bm{v}_{\zeroped{\perp}} \rvert \leq \frac{2R}{\Delta(s,t)} \} ).
\end{align}
%%%%%%%%%%%%%%%%%%%%%%%%%%%%%%%%%%%%%%%%%%%%%%%%%%%%%%%%%%%%%%%%%%%%%%%%%
In this region too,  the best upper bound that can be found for the first integral in 
(\ref{3rd contribution first expression})  is
%%%%%%%%%%%%%%%%%%%%%%%%%%%%%%%%%%%%%%%%%%%%%%%%%%%%%%%%%%%%%%%%%%%%%%%%%
\begin{gather}
W(t)-v_x  \leq \frac{C}{(1+t)}(\gamma+ A_+ \gamma^3) \notag.
\end{gather}
%%%%%%%%%%%%%%%%%%%%%%%%%%%%%%%%%%%%%%%%%%%%%%%%%%%%%%%%%%%%%%%%%%%%%%%%%
Now, $\Delta(s,t)$ can be hard to estimate.
Nevertheless so far the dominant contribution \footnote{Remember $\alpha > 1 $ so that the power of decaying in (\ref{1st contribution}) is $ 3\alpha > 3$ } in our computation is (\ref{2nd contribution}), which expresses a behaviour as $t^{-3}$ for $r_W^+$.
We can thus use a rough upper bound for the second term in 
(\ref{3rd contribution first expression}), namely
%%%%%%%%%%%%%%%%%%%%%%%%%%%%%%%%%%%%%%%%%%%%%%%%%%%%%%%%%%%%%%%%%%%%%%%%%
\begin{gather}
\int \limits_{\mathbb{R}^2 } d\bm{v}_{\zeroped{\perp}} \ e^{-\beta \bm{v}_{\zeroped{\perp}}^2} \,
				\chi ( \{\, \lvert \bm{v}_{\zeroped{\perp}} \rvert \leq \frac{2R}{\Delta(s,t)} \} )
				\leq \int \limits_{\mathbb{R}^2 } d\bm{v}_{\zeroped{\perp}} \ e^{-\beta \bm{v}_{\zeroped{\perp}}^2} \, \leq C, 
\end{gather}
%%%%%%%%%%%%%%%%%%%%%%%%%%%%%%%%%%%%%%%%%%%%%%%%%%%%%%%%%%%%%%%%%%%%%%%%%
leading to the last contribution:
%%%%%%%%%%%%%%%%%%%%%%%%%%%%%%%%%%%%%%%%%%%%%%%%%%%%%%%%%%%%%%%%%%%%%%%%%
\begin{align}
\label{3rd contribution}
&	C \int\limits_{ s(v_x) \leq s^* } dv_x \ (W(t)-v_x)^2 
	\leq C \int\limits_{ \overline{W_{t}} }^{W(t)} dv_x \ (W(t)-v_x)^2  \notag \\
&	\leq C \Bigl [\frac{A_+ \gamma^3 +\gamma}{(1+t)} \Bigl ]^3.
\end{align}
%%%%%%%%%%%%%%%%%%%%%%%%%%%%%%%%%%%%%%%%%%%%%%%%%%%%%%%%%%%%%%%%%%%%%%%%%
%%%%%%%%%%%%%%%%%%%%%%%%%%%%%%%%%%%%%%%%%%%%%%%%%%%%%%%%%%%%%%%%%%%%%%%%%
Any finer estimation regarding the second integral in (\ref{3rd contribution first expression}) would only lead to a power decay greater than $t^{-3}$,
but this would be useless right because of the dominant contribution (\ref{2nd contribution}).
Finally, collecting (\ref{1st contribution}), (\ref{2nd contribution}) and (\ref{3rd contribution}) we obtain (\ref{r_+ behaviour lemma}) for $t \geq 2s^*$.

For the remaining cases: $s^* \geq \frac{t}{2}$ and $s^* \geq t$, by the same arguments presented above, directly from (\ref{r+ v_y only}) it is straightforward to show that
\begin{gather}
r_W^+(t) \leq C \int \limits_{\overline{W_t} }^{W(t)} dv_x \ (W(t)-v_x)^2 
		\leq C \Bigl [\frac{A_+ \gamma^3 +\gamma}{(1+t)} \Bigl ]^3,
\end{gather}
which concludes the first part of the proof.\\

As far as $ r^{-}(t) $ is concerned the back of the body is simply disk shaped and we proceed in a similar way of as \cite{nostro}.

Clearly, $r_W ^-(t)=0$ for $t<t_0$ where $W$ is increasing. Instead, for $t>t_0$, the recollision condition on $ v_x $ is the same as eq.(\ref{v_x > W_t}): $v_x=\overline{W_{s,t}}<V_{\infty}$. 
Using again $ e^{-\beta v_{x}^2} -e^{- \beta v_{0x}^2}<1$ for back recollisions in eq(\ref{r-W}) we obtain
\begin{align}
r_W ^-(t) & \leq C \int _{W(t)} ^{V_{\infty}} dv_x (v_x-W)^2 = C(V_{\infty}-W)^3 \notag \\
          & \leq C \bigl ( \gamma e^{-C_+ t }+ \frac{A \gamma ^3}{(1+t)^{\alpha}} \bigr )^3 \leq C \frac{(\gamma + A \gamma ^3)^3}{(1+t)^{3 \alpha}},
\end{align}

which concludes the proof.
\end{proof}
As a direct consequence of this lemma  we have that:
\begin{gather}
\label{r+ + r- <}
r_W ^+(t)+ r_W^-(t) \leq C \frac{(\gamma + A \gamma ^3)^3}{(1+t)^{3}}
\end{gather}
which holds because $ \frac{1}{(1+t)^{3\alpha}} \leq \frac{1}{(1+t)^3}  $ for $ \alpha >1  $.
\\

We go on the next step in proving Theorem $ 1 $.
The next proposition shows the behaviour of the map $ W \to V_W $ (see \ref{map}).

\begin{prop}
\label{PROPOSITION}
Let $ W\in \Omega_\alpha $,\, then $ V_W \in \Omega_3  $.
\end{prop}

\begin{proof}
First of all, we write the solution to eq. (\ref{d dt V_W = E -FW -rW}) by means of Duhamel formula:
\begin{gather} \label{Duhamel}
V_{\infty}- V_W= \gamma e^{-\int _0 ^t K(\tau)d\tau}+ \int _0 ^t ds \ e ^{-\int _s ^t K(\tau)d\tau}(r_W^+(s) +r_W ^-(s))
\end{gather}
where 
\begin{align*}
K(t)= \frac{E-F_0(W(t))}{V_{\infty}-W(t)}.
\end{align*}
Reminding that $  V_0 \leq W(t)< V_{\infty}$ and  being $F_0$ a convex function, we get
\begin{align*}
  C_+ = F_0'(V_0) \leq K(t)   < F_0 '(V_{\infty})=C_- .
\end{align*}

By the positivity of $r_W ^+$ and $r_W ^-$, formula (\ref{Duhamel}) implies

\begin{gather}
V_{\infty}-V_W \geq \gamma e^{-\int _0 ^t K(\tau)d\tau} \geq \gamma e^{-C_- t}.
\end{gather}
which is the first property of the set $ \Omega_{\alpha}$ ; in particular we obtain $V_W (t)< V_{\infty} $ for any $t>0$.

Then, using again (\ref{d dt V_W = E -FW -rW}), we have
\begin{align*}
\frac{d}{dt}(V_W-V_0) &=K(t)(V_{\infty}-V_W )-r_W^+-r_W^- \\ 
& \geq C_+ (V_{\infty}-V_W )-C \gamma ^3   \\
&= C_+ \gamma +C_+(V_0-V_W)-C \gamma ^3 \\
& \geq -C_+(V_W-V_0)
\end{align*}
the last inequality holding for $ \gamma $ small enough; from this we get $ V_W (t) \geq V_0$ .

%
%Taken out appendix  on diff inequalities
%
%from Appendix (\ref{append diff ineq}) is :
%
%( see Appendix \ref{append diff ineq})

Moreover
\begin{align*}
\frac{d}{dt}(V_{\infty}-V_W)  \leq &-C_+(V_\infty - V_W)+r_W^+ +r_W ^-  \leq \\
&-C_+(V_\infty - V_W) + C\gamma^3 \leq \\
& C(-C_+\gamma e^{-C_-t}+ \gamma^2 ),
\end{align*}
thus the third property of $ \Omega_\alpha $
\begin{gather*}
\frac{d}{dt}(V_{\infty}-V_W) <0 \, ; \quad \forall \, t<\frac{1}{2C_-}log(\frac{C_+}{\gamma}).
\end{gather*}

In the end, from (\ref{Duhamel}) we have
\begin{gather*}
V_{\infty}- V_W \leq 
\gamma e^{-C_+ t} + \int _0 ^t ds \ e ^{-C_+ (t-s)}C 
\frac{(\gamma + A_+ \gamma ^3)^3}{(1+s)^{3}} = \\
\gamma e^{-C_+ t} + C(\gamma + A_+ \gamma ^3)^3 e ^{-C_+ t} \int _0 ^t ds  
\frac{e ^{C_+ s} }{(1+s)^3} 
\end{gather*}
We now show the following estimate:
\begin{equation*}
g(t)= e^{-C_+ t}\int _0 ^t \frac{e ^{C_+s}}{(1+s)^{3}} ds \leq \frac{C}{(1+t)^{3}}.
\end{equation*}
we can do that by proving  that: 
\begin{equation*}
g(t) \leq C  \, ; \quad \forall \, t \geq 0
\end{equation*}
and
\begin{equation*}
\lim _{t \to \infty} \frac{e^{-C_+ t} \int _0 ^t e ^{C_+s}(1+s)^{-3} ds}{(1+t)^{-3}}= C.
\end{equation*}
The first comes from direct inspection:
\begin{equation*}
g(t) \leq e^{-C_+ t}e^{C_+ t}\int _0 ^t \frac{1}{(1+s)^{3}} ds  \leq   \int _0 ^{\infty} \frac{1}{(1+s)^{3}}=C 
\end{equation*}

The second can be easily computed, for instance by  De L'Hopital:
\begin{equation*}
\lim _{t \to \infty} \frac{\int _0 ^t e ^{C_+s}(1+s)^{-3} ds}{ e^{C_+ t}  (1+t)^{-3}}=
 \lim _{t \to \infty}\frac{1}{(C_+-\frac{3}{1+t})} = C
\end{equation*}

To conclude there exists a constant $ \bar{C} $ :
\begin{equation*}
V_{\infty}- V_W \leq \gamma e^{-C_+ t}+ \bar{C}\frac{(\gamma + A_+ \gamma ^3)^3}{(1+t)^{3}}. 
\end{equation*}
Now, to obtain the second property of $ \Omega_\alpha $ ,  with $ \alpha =3  \,$ it is sufficient that
\begin{equation}
\bar{C}(\gamma + A_+ \gamma ^3)^3 <A_+ \gamma ^3 \notag
\end{equation}
The last inequality is satisfied by choosing  $A_+=2\bar{C}$ (this fixes the constant $A_+$) and $\gamma$ consequently small.
\end{proof}

We can now follow the proof in \cite{E=0} to show the existence and hence to prove Theorem $ 3.1 $.  We showed in proposition \ref{PROPOSITION}  that 
$ \mathcal{F} \, : \Omega_3 \to  \Omega_3 $, where $ \mathcal{F} $ is the map (\ref{map}). Consider the set $ \mathcal{K} $ of  functions $ W \in \Omega_3 $ enjoying the property:
\begin{gather}
\mathcal{K} = \{ W\in \Omega_3 \, \big |  \, ess\, sup_{t \in \mathbb{R^+}}(|W(t)|+|\dot{W}(t)|)=L < \infty \}
\end{gather}
It results that $ \mathcal{K} $ is compact and convex  .
The map (\ref{map}) $ W \to V_W  $ from $ \mathcal{K} $ to itself is a continuous map as showed  in \cite{nostro} ( to which we refer only to show the continuity of the map, being the existence proof there not complete), then by the Schauder fixed point theorem 
$ \mathcal{F} $ has at least a fixed point in $ \mathcal{K} $, that is our seeked solution.

To conclude the proof of Theorem 3.1  we consider any solution $ (V,f) $ of the problem (\ref{vlasov})-(\ref{d/dt V = E- F(t)}).
By continuity there exists a time interval for which
\begin{gather}
\label{ogni soluzione ha ..}
V_{\infty}- V(t) < \gamma e^{-C_+ t}+ C\frac{A_+}{(1+t)^{3}}\gamma ^3, 
\end{gather}
as it is clearly valid at time zero.
Let $ T $ be the first time this inequality is violated,
by the same arguments of Proposition (\ref{PROPOSITION}) (replacing $ W $ by $ V $) we have:
\begin{gather}
V_{\infty}-V(t) \geq \gamma e^{-C_- t}.
\end{gather}
and 
\begin{gather}
\frac{d}{dt} (V_{\infty}-V(t) ) \leq 0.
\end{gather}
for $ t \in [0,min(t_0,t)] $.
Since $ V $ enjoys the same properties as $ W $ for $ t \in [0,T)  $ we infer that eq.(\ref{ogni soluzione ha ..})  holds globally in time.
This concludes the proof of Theorem 3.1.

%%%%%%%%%%%%%%%%%%%%%%%%%%%%%%%%%%%%%%%%%%%%%%%%%%%%%%%%%%%%%%
%%%%%%%%%%%%%%%%%%%%%%%%%%%%%%%%%%%%%%%%%%%%%%%%%%%%%%%%%%%%%%%%%%%%
%%%%%%%%%%%%%%%%%%%%%%%%%%%%%%%%%%%%%%%%%%%%%%%%%%%%%%%%%%%%%%
%%%%%%%%%%%%%%%%%%%%%%%%%%%%%%%%%%%%%%%%%%%%%%%%%%%%%%%%%%%%%%%%%%%%
%%%%%%%%%%%%%%%%%%%%%%%%%%%%%%%%%%%%%%%%%%%%%%%%%%%%%%%%%%%%%%
%%%%%%%%%%%%%%%%%%%%%%%%%%%%%%%%%%%%%%%%%%%%%%%%%%%%%%%%%%%%%%%%%%%%
%%%%%%%%%%%%%%%%%%%%%%%%%%%%%%%%%%%%%%%%%%%%%%%%%%%%%%%%%%%%%%
%%%%%%%%%%%%%%%%%%%%%%%%%%%%%%%%%%%%%%%%%%%%%%%%%%%%%%%%%%%%%%%%%%%%

\subsection{Improvement of the lower bound}
\label{SECTION improvement}

Before starting our analysis we remark that following the same steps taken in Section \ref{SECTION kinematic issues} (replacing $ W $ by $ V $), we arrive at:
\begin{gather}
s_- (\gamma) \leq s^* \leq s_+ (\gamma).
\end{gather}
with
\begin{gather}
s_- (\gamma) = \frac{1}{C_-} log \frac{\gamma}{hC_-} \, ; \quad  s_+ (\gamma) = \sqrt{\frac{C \gamma}{h}}-1
\end{gather}

We present here the strategy  to prove Theorem 3.2.
Theorem 3.1 proved that our solution is bounded as
 \begin{gather*}
 \gamma e^{-C_- t} \leq V_{\infty}- V(t) \leq 
 \gamma e^{-C_+ t}+ \frac{A_+}{(1+t)^{3}}\gamma ^3, 
\end{gather*}

by usual comparison argument, using that $ r^-(t) \geq 0 $,
\begin{gather}
\frac{d}{dt}(V_{\infty}- V(t)) \geq -C_- (V_{\infty}- V(t)) + r^+(t),
\end{gather}
then by Duhamel expression 
\begin{gather}
\label{dhuam finale}
V_{\infty}- V(t) \geq \gamma e^{-C_- t} + \int_ {0}^{t} ds e^{-C_-(t-s) }r^+(s)
\end{gather}

Exploiting $ r^+(t) >0   $ sends us back to the main properties shown in Theorem 3.1.
We thus look for a finer lower bound to $ r^+(t)  $ for large $ t $; in order to achieve this new bound we integrate over velocities producing a single recollision in the past.

Let's consider the recollision condition along x-axis ($ v_x$ is the velocity of the particle before the impact at time $ t $):
\begin{gather*}
\exists \, s \leq t \, : \quad  v_x= \overline{V_{s,t}}.
\end{gather*} 
A further recollision at $ \tau < s $ happens if (the velocity of the particle before the impact at $ s $ is $ 2V(s) -v_x $)
\begin{gather}
\exists \, \tau \leq s \, : \quad  2V(s) -v_x = \overline{V_{\tau,s}},
\end{gather}
now, $ \overline{V_{s,t}} >V_0 $, therefore in order to have only one recollision a sufficient condition is 
\begin{gather*}
\exists \, s \leq t \, : \quad  v_x= \overline{V_{s,t}}. \\
2V(s) -v_x  < V_0
\end{gather*}
In a compact formulation, the region of single recollision times is 
\begin{gather}
\{s \in (0,t)  \, : \quad q(s)<0  \}, \\
q(s) = 2V(s)  - \overline{V_{s,t}} - V_0,
\end{gather}
evidently
\begin{gather}
q(0) = V_0 -\overline{V_{t}} <0 \\
q(t) = V(t)-V_0 >0,
\end{gather}
then there exists $ s_0 > 0 $, the smallest solution to $ q(s_0) = 0 $ :
\begin{gather}
\label{def s_0} 
V(s_0)=\frac{\overline{V_{s_0,t}}+V_0}{2} .
\end{gather}
so that:
\begin{gather}
q(s)<0  \, : \quad \forall \, 0< s < s_0.
\end{gather}

Before estimating recollision terms we  show that, for $t$ sufficiently large and $\gamma$ sufficiently small,
%%%%%%%%%%%%%%%%%%%%%%%%%%%%%%%%%%%%%%%%%%%%%%%%%%%%%%%%%%%%%%%%%%%%%%%%%%%%%%%%%%% 
\begin{align}
m_1 \leq s_0 \leq m_2,
\end{align} 
where $m_1=\frac{1}{C-}log \frac{3}{2}$ and $m_2=\frac{1}{C+}log4$.

The general property of the solution computed at $ s_0 $ reads
\begin{gather}
\gamma e^{-C_- s_0}  \leq  \ V_{\infty}-V(s_0) \  \leq \gamma e^{-C_+s_0} +\frac{A \gamma ^3}{(1+s_0)^3}
\end{gather}
where, using (\ref{def s_0}), 
\begin{gather}
V_{\infty}-V(s_0) = \frac{\gamma}{2} + \frac{V_{\infty} - \overline{V_{s_0,t}}}{2}.
\end{gather}
Hence, for $\gamma$ sufficiently small
\begin{align}
\gamma e^{-C_+ s_0} & \geq \frac{\gamma}{2} + \frac{V_{\infty} - \overline{V_{s_0,t}}}{2} - \frac{A \gamma ^3}{(1+s_0)^3} \notag \\ & \geq \frac{\gamma}{2} - \frac{A \gamma ^3}{(1+s_0)^3}  \geq \frac{\gamma}{2} - A \gamma ^3 \geq \frac{\gamma}{4}, \notag
\end{align}
leading to $ m_2=\frac{1}{C+}log4  $.\\
For the lower bound,
\begin{align*}
& \gamma e^{-C_-s_0}  \leq \frac{\gamma}{2} +\frac{V_{\infty} -\overline{V_{s_0 t}} }{2} = \\
& = \frac{\gamma}{2} +\frac{1}{2(t-s_0)} \int_{s_0}^t V_{\infty}-V(\tau) \ d\tau  \\
 & \leq \frac{\gamma}{2} + \frac{1}{2(t-s_0)} \int_{0}^{\infty} \gamma e^{-C_+ \tau} +\frac{A \gamma ^3}{(1+ \tau)^3} \ d\tau  \\ 
 & \leq \frac{\gamma}{2} + C \frac{\gamma +A \gamma^3}{2(t-s_0)} \\
 &  \leq \frac{\gamma}{2} + C \frac{\gamma +A \gamma^3}{2(t-m_2 )}  \leq \frac{2}{3} \gamma , 
\end{align*}
which holds for a fixed $\gamma$ sufficiently small and $t$ sufficiently large, and the proof is concluded. \\ 
%%%%%%%%%%%%%%%%%%%%%%%%%%%%%%%%%%%%%%%%%%%%%%%%%%%%%%%%%%%%%%%%%%%%%
So the desired lower bound is expressed by:
\begin{gather}
r^+(t) =C \int\limits _{ | \bm{x}_{\zeroped{\perp}} | \leq R } d\bm{x}_{\zeroped{\perp}} 
\int\limits_{\overline{V_t}} ^{V(t)} dv_x (v_x -V(t))^2 
\int \limits_{\mathbb{R}^2} d\bm{v}_{\zeroped{\perp}} e^{- \beta \bm{v}_{\zeroped{\perp}}^2}(e^{- \beta v_{0x}^2} - e^{- \beta v_x^2}) \geq  \notag \\
\label{r+ >=}
\int\limits_{ | \bm{x}_{\zeroped{\perp}} | \leq R } d\bm{x}_{\zeroped{\perp}}  
\int\limits_{\overline{V_t}} ^{\overline{V_{s_0,t}}} dv_x (v_x -V(t))^2 
\int \limits_{\mathbb{R}^2} d\bm{v}_{\zeroped{\perp}} e^{- \beta \bm{v}_{\zeroped{\perp}}^2}(e^{- \beta v_{0x}^2} - e^{- \beta v_x^2}).
\end{gather}

In the following we will prove these inequalities:
\begin{gather}
\label{e-e > gamma}
(e^{- \beta v_{0x}^2} - e^{- \beta v_x^2}) \geq C \gamma \\
\label{V(t)-v_x>}
V(t)-v_x \geq C \frac{\gamma}{t} \\
\label{V_(s0 t) - V_(t) >}
\overline{V_{s_0,t}} -\overline{V_t} > C \frac{\gamma}{t}
\end{gather}
First, setting $ j(v)=  e^{- \beta v } $, we can write
\begin{gather*}
(e^{- \beta v_{0x}^2} - e^{- \beta v_x^2}) = -( j(v_x^2)- j(v_{\zeroped{0}x}^2) ) =  \\
-j'(\eta) (v_x ^2 - v_{\zeroped{0}x}^2) = |j'(\eta)| (v_x ^2 - v_{\zeroped{0}x}^2)
 \, ; \quad \eta \in [v_{\zeroped{0}x}^2,v_x^2] 
\end{gather*}
 Since $v_x^2$ and $v_{\zeroped{0}x}^2= (2V(s)-v_x)^2$ are bounded,   $  |j'(\eta)|  $ is bounded in $ \eta$, hence
\begin{gather}
(e^{- \beta v_{0x}^2} - e^{- \beta v_x^2}) \geq C (v_x ^2 - v_{\zeroped{0}x}^2)  
\end{gather}
Now, if $ \bm{v}_{\zeroped{\perp}} $ doesn't bring recollision $ v_{\zeroped{0}x} \equiv v_x $
and we are back to $  (e^{- \beta v_{0x}^2} - e^{- \beta v_x^2}) > 0 $, thus provided that $ \bm{v}_{\zeroped{\perp}} $ belongs to recollision region we have
\begin{gather*}
(v_x ^2 - v_{\zeroped{0}x}^2) = v_x^2 -(2V(s)-v_x)^2 = \\
C V(s)(v_x-V(s))> C V(0)(v_x -V(s)),
\end{gather*}
in the region $ 0<s < s_0 $ it yields 
\begin{gather*}
v_x-V(s) \geq v_x -\frac{V_0+\overline{V_{s,t}} }{2}=\frac{\overline{V_{s,t}}-V_0 }{2} 
\end{gather*}
leading to
\begin{gather}
(e^{- \beta v_{0x}^2} - e^{- \beta v_x^2}) \geq C (\overline{V_{s,t}} - V_0) ; 
\end{gather}
we compute the last term as follows
\begin{align*}
 \overline{V_{s,t}}-V_0  \geq & \overline{V_{t}}-V_0                                                                               = \frac{1}{t} \int _0 ^t V(\tau)-V_0 \ d\tau = \\
      &   = \gamma  -\frac{1}{t} \int _0 ^t V_{\infty} -V(\tau) d \tau  \\
 &  \geq \gamma -\frac{1}{t} \int _0 ^t \gamma e^{-C_+{\tau}} +\frac{A_+ \gamma ^3}{(1+\tau)^3} d\tau  \\
 & \geq \gamma -\frac{1}{t} \int _0 ^{\infty} \gamma e^{-C_+{\tau}} +\frac{A_+ \gamma ^3}{(1+\tau)^3} d\tau \notag \\
 &  \geq \gamma -\frac{C}{t} (\gamma + A_+ \gamma ^3) \geq \frac{\gamma}{2} 
\end{align*}
the last inequality holding for fixed $\gamma$ and $t$ sufficiently large, so that 
(\ref{e-e > gamma}) is proved.
\\
We proceed in computing   (\ref{V(t)-v_x>}) :
\begin{align*}
V(t)-v_x  & = V(t)-\overline{V_{s,t}} = \frac{1}{t-s} \int_s ^t V(t)-V(\tau)\  d \tau  \\
 &\geq \frac{1}{t} \int_{m_2} ^t V(t)-V(\tau) \ d \tau  \\ 
 & = \frac{1}{t} \int_{m_2} ^t V(t)- V_{\infty}+V_{\infty}-V(\tau) \  d\tau  \\
 & \geq  \frac{1}{t} \int_{m_2} ^t -\gamma e^{-C_+ t} -\frac{A \gamma^3}{(1+t)^3} +\gamma e^{-C_- \tau} d\tau  \\
 & \geq   \frac{1}{t} (t-m_2) (-\gamma e^{- C_+ t} -\frac{A \gamma^3}{(1+t)^3}) -\frac{\gamma}{t} \frac{e^{-C_- t}}{C_-} + \frac{e^{-C_-m_2}}{C_-} \frac{\gamma}{t}  \\
 & \geq \frac{e^{-C_-m_2}}{2C_-} \frac{\gamma}{t} = C \frac{\gamma}{t} 
\end{align*}
the last inequality holding for large $t$, given a fixed $\gamma$.
%%%%%%%%%%%%%%%%%%%%%%
For the last expression first we notice that it can be written as
\begin{gather}
\overline{V_{s_0, t}}-\overline{V_{t}}=\frac{s_0}{t-s_0} ( \overline{V_t}-\overline{V_{s_0}})
\end{gather}
Since $s_0$ is bounded, we note that, for $\gamma$ small, $t_0$ is much larger than $s_0$, so that $V(\tau)$ is increasing for $\tau \leq s_0$. Hence
\begin{align}
  \overline{V_t} - \overline{V_{s_0}}   & =  \frac{1}{t} \int_ 0 ^t V(\tau)d\tau-\frac{1}{s_0} \int_ 0 ^{s_0} V(\tau)d\tau \notag \\
   & \geq \frac{1}{t} \int_ 0 ^t V(\tau)d\tau - V(s_0) \notag \\
   &  = \frac{1}{t} \int_ 0 ^t \Bigl ( V(\tau)-V_{\infty} \Bigr ) d\tau  +V_{\infty}-V(s_0) \notag 
\end{align}    
By using relation (\ref{def s_0}), we observe that 
\begin{equation}
V_{\infty}-V(s_0)= V_{\infty} - \frac{\overline{V_{s_0,t}}+V_0 }{2} \geq V_{\infty} - \frac{V_{\infty}+V_0 }{2}= \frac{\gamma}{2} \notag
\end{equation}
Hence
\begin{gather*} 
\overline{V_t} - \overline{V_{s_0}}      \geq \frac{1}{t} \int_{0}^t ( - \gamma e^{-C_+ \tau} - \frac{A \gamma ^3}{(1+\tau) ^ 3})d\tau +\frac{\gamma}{2} \notag \\ 
                                     \geq \frac{1}{t} \int_{0}^{\infty} ( - \gamma e^{-C_+ \tau} - \frac{A \gamma^3}{(1+\tau) ^3})d\tau +\frac{\gamma}{2} \notag \\
   \geq -C \frac{\gamma + A \gamma ^ 3}{t} + \frac{\gamma}{2}  \geq \frac{\gamma}{4},
\end{gather*}
for $t$ large enough, and finally arrive to (\ref{V_(s0 t) - V_(t) >})
\begin{equation}
\overline{V_{s_0, t}}-\overline{V_{t}}= \frac{s_0}{t-s_0} ( \overline{V_t}-\overline{V_{s_0}}) \geq C \frac{\gamma}{t}.
\end{equation}

At this point we turn our attention to recollision terms; we first analyze the case in which no positive $s^*$ exists. It happens whenever $\gamma$ is so small that $s_+(\gamma)<0$ i.e $\gamma < \frac{h}{C} $. In this case no particles can leave the box, thus $ \bm{v}_{\zeroped{\perp}} $ ranges over all $\mathbb{R}^2$, and by means of inequalities (\ref{e-e > gamma}), (\ref{V(t)-v_x>}) we can finally compute 
(\ref{r+ >=}) :
%%%%%%%%%%%%%%%%%%%%%%%%%%%%%%%%%%%%%%$s^*(\gamma)<0$%%%%%%%%%%%%%%%%%%%%%%%%%%%%%%%%%%%%%%%%%%%%%
\begin{gather*}
r^+(t)\geq \int\limits_{ | \bm{x}_{\zeroped{\perp}} | \leq R } d\bm{x}_{\zeroped{\perp}}  
\int\limits_{\overline{V_t}} ^{\overline{V_{s_0,t}}} dv_x (v_x -V(t))^2 
\int \limits_{\mathbb{R}^2}  d\bm{v}_{\zeroped{\perp}} e^{- \beta \bm{v}_{\zeroped{\perp}}^2}(e^{- \beta v_{0x}^2} - e^{- \beta v_x^2}) \geq  \notag \\
C \int\limits_{\overline{V_t}} ^{\overline{V_{s_0,t}}} dv_x (v_x -V(t))^2 
\int\limits_{\mathbb{R}^2}  d\bm{v}_{\zeroped{\perp}} e^{- \beta \bm{v}_{\zeroped{\perp}}^2}\gamma \geq  \notag \\
C\gamma \int\limits_{\overline{V_t}} ^{\overline{V_{s_0,t}}} dv_x (v_x -V(t))^2 \geq 
C\gamma ( \overline{V_{s_0,t}}- {\overline{V_t}})  \bigl (\frac{\gamma}{t} \bigl)^2,
\end{gather*}
exploiting (\ref{V_(s0 t) - V_(t) >}) we conclude that, for $ t $ sufficiently large, independently of $ \gamma $:
\begin{equation}
r^+(t)\geq C \frac{\gamma ^4}{t^3}
\end{equation}
It is true that for any given $h$ there is always a $ \gamma $ small enough so that $s_+(\gamma)<0$, but physically  this is would be the less interesting case in which the particle is not allowed to escape the box, in other words it would be the same case of a box with endless barriers.

We therefore prove our estimate dropping the hypothesis of $s_+(\gamma)<0$ and  consider the case in which a positive $s^*$ exists, provided smaller than $s_0$.
The last condition is possible by choosing $s_+(\gamma) < m_1/2  $ that is:
$ \frac{h}{C} \leq \gamma \leq \frac{h}{C}(1+ m_1/2)^2$.

We can now compute the lower bound for $ r^+(t) $ in the region of $ s \in (s^*,s_0) $ where there is again no lower bound for $ \bm{v}_{\zeroped{\perp}} $.
%%%%%%%%%%%%%%%%%%%%%%%%%%%%%%%%%%%%%%%%%%%%%%%%%%%%%%%%%%%%%%%%%%%%%%%%%%%%%%%%%%%
\begin{gather}
r^+(t)\geq \int\limits_{ | \bm{x}_{\zeroped{\perp}} | \leq R} d\bm{x}_{\zeroped{\perp}}
 \int\limits_{\overline{V_{s^{*},t}}} ^{\overline{V_{s_0,t}}} dv_x (v_x -V(t))^2 
 \int \limits_{\mathbb{R}^2} d\bm{v}_{\zeroped{\perp}} e^{- \beta \bm{v}_{\zeroped{\perp}}^2}(e^{- \beta v_{0x}^2} - e^{- \beta v_x^2}) .
\end{gather}
%%%%%%%%%%%%%%%%%%%%%%%%%%%%%%%%%%%%%%%%%%%%%%%%%%%%%%%%%%%%%%%%%%%%%%%%%%%%%%%%%%
In this case we have to find a lower bound for the new term
$ \overline{V_{s_0, t}} - \overline{V_{s^*,t}}$.\\
At first, we notice that
\begin{align}
&\overline{V_{s_0, t}}-\overline{V_{s^*,t}}= \notag \\ 
	&=\frac{1}{t-s_0}\int_{s_0}^t V(\tau)d\tau - 
	\frac{1}{t-s^*} \int_{s^*}^t V(\tau)d\tau \notag \\  
 	&= \frac{1}{t-s_0}\int_{s^*}^t V(\tau)d\tau -
 	\frac{1}{t-{s_0}} \int _{s^*}^{s_0} V(\tau)d\tau -
	 \frac{1}{t-s^*} \int _{s^*} ^t V(\tau)d\tau \notag \\ 
  	&=\frac{s_0 -s^*}{(t-s_0)(t-s^*)} \int _{s^*} ^t V(\tau)d\tau -\frac{1}{t-s_0}\int_{s^*}^{s_0} V(\tau)d\tau \\
  	  &=\frac{s_0-s^*}{t-s_0} ( \overline{V_{s^*,t}}-\overline{V_{s^*,s_0}}), \notag \\
\end{align}
 we have that 
\begin{align}
\frac{s_0-s^*}{t-s_0} \geq \frac{m_1}{2t},
\end{align}
since $s_0 \geq m_1$ and $s^* \leq s_+ \leq \frac{m_1}{2}$, while, for t sufficiently large, 
 \begin{align}
\label{aa} (\overline{V_{s^*,t}} - \overline{V_{s^*,s_0}}) \geq \frac{\gamma}{4} .
\end{align}
We now prove the last inequality.

First of all, since $s_0$ is bounded, we note that, for $\gamma$ small, $t_0$ is much larger than $s_0$, so that $V(\tau)$ is increasing for $\tau \leq s_0$. Hence
\begin{align*}
  & \overline{V_{s^*,t}} - \overline{V_{s^*,s_0}} =  \\
  & =\frac{1}{t-s^*} \int_{s^*}^t V(\tau)d\tau -\frac{1}{s_0 - s^*} \int_{s^*}^{s_0} V(\tau)     d\tau  \\
 &  \geq \frac{1}{t-s^*} \int_{s^*}^t V(\tau)d\tau - V(s_0) \notag \\
 & \geq \frac{1}{t-s^*} \int_{s^*}^t ( - \gamma e^{-C_+ \tau} - \frac{A \gamma^3}{(1+\tau) ^3})d\tau +V_{\infty} -V(s_0)  \\
    & \geq \frac{1}{t-s^*} \int_{0}^{\infty} ( - \gamma e^{-C_+ \tau} - \frac{A \gamma^3}{(1+\tau) ^3})d\tau +\frac{\gamma}{2}  \\
  &   \geq \frac{\gamma}{4},
\end{align*}
for $t$ large enough. \\
Therefore we obtain 
\begin{align}
\label{estremi} 
\overline{V_{s_0, t}}-\overline{V_{s^*,t}} \geq \frac{C \gamma}{t}.
\end{align}

In the end even in the case of $ s_+(\gamma) > 0 $, proceeding as before, we conclude that for
 $ t $ sufficiently large, say $ t > {\overline{t}} $ independent of $ \gamma $,
\begin{align}
\label{r end}
r^+(t) \geq C \frac{\gamma ^4}{t^3}.
\end{align}

Eventually the expression (\ref{dhuam finale}) is 
\begin{gather}
V_{\infty}- V(t) \geq \gamma e^{-C_- t} + \int_ {0}^{t} ds e^{-C_-(t-s) }r^+(s)
\geq \\ 
\gamma e^{-C_-t}+ C\int _{\overline{t}} ^t ds \ e ^{-C_-(t-s)} \frac{\gamma ^4}{s^3} \notag
\end{gather}

and we can bound  the integral by showing that:
\begin{equation}
\lim _{t \to \infty} \frac{\int _{\overline{t}} ^t ds \ e ^{-C_-(t-s)} s^{-3} }{(1+t)^{-3}}=C_b \notag
\end{equation}

where $ C_b $ is a constant, we finally obtain for $t>\overline{t }$, with $ \overline{t} $ sufficiently large, independent of $ \gamma $
\begin{equation}
V_{\infty}- V\geq \gamma e^{-C_-t}+ A_-\frac{\gamma ^4}{t^3} \notag
\end{equation}
which concludes the proof of Theorem 3.2.

\section{ Comments }
In this work we investigated the viscous friction acting on a body of concave shape immersed in a free gas.

We studied  the free Vlasov equation by means of characteristics, coupled with the ordinary differential equation describing the velocity of the body.
We split the action of the gas on the body, namely $ F(t) $, in the contribution
$ r^{\pm}(t) $ coming from recollisions.

The full problem was presented in Section 3. We studied it as a fixed point of the map (\ref{map}); we stress that our techniques, as those of the other articles on this subject, are perturbative and work only when the parameter
 $ \gamma =V_{\infty}- V_0 $ is finite but sufficiently small.
 
A key role in the analysis was played by the term $ s^* $ in (\ref{f(s*= h}).
This value was linked through Lemma 3.2 with the time that particles spent outside the body;
proving that it was bounded in time we essentially proved that there was a non zero measure of particles trapped in the cylinder all along its evolution (at least for $  \gamma $  sufficiently small) and this was essentially responsible for the slower rate of decay to the limiting velocity.

The estimate of $r^{\pm}(t)$ was carried out in detail in order to keep track of each contribution. In the end, in order to obtain the improvement on the lower bound of 
Theorem 3.2, we integrated the recollision term over particles producing a single recollision, indeed the main contribution to the friction comes from these very ones.\\

We remark that if the cylinder was placed the other way round, that is with its hollow base facing backwards, we could proceed as in \cite{nostro} for $ r^+ $ and in the same way of section (\ref{SECTION estimate r+}) for $ r^- $ leading to the usual $ t^{-5} $ power of decay.
It is natural, in fact, that the trapping effects vanish in the case of a cavity not pointed towards the motion of the body.
  
Some possible generalizations  of the present paper can be considered,
for example the case in which the external field is absent, namely $E=0$.  
As in \cite{E=0}, the evolution of the system is different from what intuition may suggest: assuming an initial value $V_0>0$, the velocity  reaches zero in a finite time $\overline{t}$ , then becomes negative  (thanks to recollision terms), and  finally it reaches zero from negative values.  Regarding the exact trend in time of the solution there are two different cases; if the concavity  is turned to the negative x-direction (hollow base facing backwards), we expect a $t^{-3}$ decay. Here is a sketch of the  proof: the contribution given by the back recollisions is bounded, for large $t$, in the following way:
   \begin{align*}
C_1 \frac{\gamma ^5}{t^3}\leq r_W^-(t)\leq C_2 \frac{(\gamma + A_1 \gamma ^3)^3}{(1+t)^3} 
\end{align*}
while the frontal recollision term can be estimated as
\begin{align*}
|r_W^+(t)| \leq C \frac{\gamma^{5+\frac{1}{4}}A_1^3}{(1+t)^5} \chi (t>\overline{t}).
\end{align*}
Clearly $r_W^-(t)$ is dominant and, using Duhamel formula, produces the $t^{-3}$ decay.
 On the other hand, if the concavity is turned to the positive x-direction, although there is a sort of trapping effect at large times , we are not able to express ourselves over the exact asymptotic behavior and a further investigation is needed.\\
 Another physically interesting case is the one in which  $V_0>V_{\infty}$. In this case we expect that $V(t)$ initially decreases and crosses the limiting value $V_{\infty}$ within a finite time $t_0$, then it reaches $V_{\infty}$ from below, with a $t^{-3}$ trend; indeed, there exists a time $s^*>t_0$ such that a fraction of particles colliding at $s>s^*$ remain trapped in the concavity. However, an estimate of the magnitude of $s^*$ depending on $\gamma$ and $h$ should require a finer analysis.
\\
 Another interesting case is that of a stochastic interaction between the concave body and the gas. More precisely, as already done in previous papers, we can assume that colliding particles are absorbed and immediately re-emitted with Maxwellian-distribuited velocities (see \cite{diffusivo}, \cite{strauss}, \cite{STRAUSS-2}):
\begin{align} \label{diffusive}
f_+(\textbf{x},\textbf{v},t)=\alpha J(\textbf{x},t) e ^{-\beta (v_x-V(t))^2}e^{-\beta v_\perp ^2}
\end{align}

 As a first study, one could impose diffusive boundary conditions only on the bottom of the hollow cylinder; the side, instead, could be still assumed adiabatic, so it simply acts as a lateral barrier which reflects elastically the gas particles. 
 Regardless of their colliding velocity, particles are now re-emitted, with high probability, at a velocity near to $V(t)$ (see (\ref{diffusive})) and this clearly favours their trapping due to lateral barriers. We emphasize that even particles colliding  at high speed, which would be swept away in the case of elastic collisions, can be subject to the trapping effect.
 So, the number of recollisions and the portion of trapped particles increases with respect to the elastic case and we expect that it directly influences the asymptotic behavior.
\\

As stressed in the introduction the aim of this work was to get a slightly more realistic perspective in the study of such friction problems, taking into account simple interactions between the body shape and the gas, another work in this direction is \cite{altro}, where the body is considered elastic and changes its length according to the interaction with particles. 
The next step would be to consider  more general kinds of concavity.   Nevertheless we would like to stress that even slight generalizations can give  rise to further complications. Consider for example the case of a body with tilted walls: even if constrained to a fixed velocity, the gas particles would bounce inside correlating through times different areas of the inner side of the body (in the cylinder case collisions with the lateral barriers didn't change particles momentum along the x axis); this gives rise to a time dependent friction for constant velocities (which was absent in all previous cases).
The time dependency of the friction is in fact linked to the recollision terms and in the case of fixed velocitiy for the body  all previous scenarios showed no recollisions, however this would not be the case anymore if the body presented tilted walls.

\appendix 
\section{Appendix: Gas Particles Kinematic}
\label{appendix collision}

We derive here collision conditions (\ref{v'}) and (\ref{v tilde}).
In what follows we denote by $M$ and $V$ the mass and velocity of the body and by $m$ and $\bm{v}$ mass and velocity of a particle which will be assumed to collide elastically with the body.
We start with lateral collisions.\\
The expression of the side surface of our body, and its inner normal are
\footnote{ $\bm{ \nabla_{\perp}  } = (\partial _y, \partial _z)  $}:
\begin{gather}%%%%%%%%%%%%%%%%%%%%%%%%%%%%%%%%%%%%%%%%%%%%%%%%%%%%%%%%%%%%%%%%%%%%%%%%%
g(\bm{x_{\perp}}) = |\bm{x_{\perp}}|^2- R^2 = 0 \\
\hat{\bm{n}_S} =(0 , - \bm{ \nabla_{\perp}  }\,g )\frac{1}{| \bm{ \nabla_{\perp}  }\,g |} =  
(0 , - \bm{x_{\perp} } )\frac{1}{R}.
\end{gather}%%%%%%%%%%%%%%%%%%%%%%%%%%%%%%%%%%%%%%%%%%%%%%%%%%%%%%%%%%%%%%%%%%%%%%
Let now  $\bm{F}^b $ be the impulsive force acting on $ S(t)  $ during a collision with a gas particle, conversely let $\bm{F}^g $  be the impulsive force that the gas particle undergoes during this collision, of course by Newton third law they are opposites. Moreover, being a constraint force, $\bm{F}^g $ is perpendicular to the surface
(in what follows the upper sign will refer to inner collisions, the lower to outer ones):
\begin{gather}%%%%%%%%%%%%%%%%%%%%%%%%%%%%%%%%%%%%%%%%%%%%%%%%%%%%%%%%%%%%%%%%%%%%%%
\label{azione reazione}
\bm{F}^b(s)= - \bm{F}^g(s) \\
\bm{F}^g = \pm |\bm{F}^g| \hat{\bm{n}}_S
\end{gather}%%%%%%%%%%%%%%%%%%%%%%%%%%%%%%%%%%%%%%%%%%%%%%%%%%%%%%%%%%%%%%%%%%%%%%

In particular we consider that the particle hits the body  in $ (x,\bm{x_{\perp} }) \in  S(t) $ at time $ t $ and define its pre(post)collisional velocity as $ \bm{v} $ $  ( \widetilde{\bm{v}} )$;
the momentum change of the body $(\bm{P}^b)$ due to this collision is:
\begin{gather}%%%%%%%%%%%%%%%%%%%%%%%%%%%%%%%%%%%%%%%%%%%%%%%%%%%%%%%%%%%%%%%%%%%%%%%
\frac{d \bm{P}^b}{ds}(s)= \mp |\bm{F}^g(s)| \hat{\bm{n}_S} 
\end{gather}%%%%%%%%%%%%%%%%%%%%%%%%%%%%%%%%%%%%%%%%%%%%%%%%%%%%%%%%%%%%%%%%%%%%%%
and by components:
\begin{gather}%%%%%%%%%%%%%%%%%%%%%%%%%%%%%%%%%%%%%%%%%%%%%%%%%%%%%%%%%%%%%%%%%%%%%%
\label{dPx^b = 0}
\frac{d P^{b}_{x} }{ds}(s)= 0 \\
\frac{d \bm{P}^{b}_{\perp}}{ds}(s)= - \bm{F}^g(s)= \pm |\bm{F}^g(s)|\, \bm{x_{\perp} } \frac{1}{R}.
\end{gather}%%%%%%%%%%%%%%%%%%%%%%%%%%%%%%%%%%%%%%%%%%%%%%%%%%%%%%%%%%%%%%%%%%%%%%

%(\ref{v tilde}).
Even if the second of these equations represents a variation of momentum along  $ \bm{x}_{\perp} $ axis due to a single collision, the gas distribution is invariant in this direction (the system clearly possesses such invariance) acting in an homogeneous way around the surface $ S(t) $, therefore the overall action of the gas has a null effect on  the momentum   along  the $ x_{\perp} $ axis \footnote{In any case the body is supposed constrained along the $x$ direction.} .

On the other hand the system lacks symmetry along x-axis (the very displacement of the cylinder being in this direction), in this sense eq (\ref{dPx^b = 0}) guarantees that lateral collisions don't change the velocity of the body.

Focusing our attention on the particle, by virtue of (\ref{azione reazione}) we have:
\begin{gather}
\frac{d P^{g}_{x} }{ds}(s)= 0 \\
\frac{d \bm{P}^{g}_{\perp}}{ds}(s)= \mp |\bm{F}^g(s)|\, \bm{x_{\perp} } \frac{1}{R}
\end{gather}
and integrating the above relations in an interval of time $ [t-\epsilon, t+\epsilon] $ such that it contains only the collision at  $ t $, we obtain the total variation:
\begin{gather}
\label{append vx tilde-vx =0}
m(\widetilde{v_x}-v_x)=0 \\
\label{append m (vperp tilde -v)}
m( \widetilde{\bm{v}}_{\perp} -\bm{v_{\perp} } )=\mp \frac{I}{R} \bm{x_{\perp} } \ ; \quad 
I = \int_{t-\epsilon}^{t+ \epsilon} |\bm{F}^g(s)|\, ds,
\end{gather}
To complete the derivation we use the conservation of kinetic energy of the system, remembering that the body doesn't change its velocity for lateral collision and using eq.(\ref{append vx tilde-vx =0}):
\begin{gather}
\label{append |vperp tilde| =|vperp|)}
|\widetilde{\bm{v}}_{\perp}| =|\bm{v}_{\perp}|
\end{gather}
Finally taking the norm of eq.(\ref{append m (vperp tilde -v)}) and exploiting eq.(\ref{append |vperp tilde| =|vperp|)}) we get:
\begin{gather}
|\widetilde{\bm{v}}_{\perp}|^2 = |\bm{v}_{\perp} \mp \frac{I}{mR}\bm{x}_{\perp} |^2 =
 |\bm{v}_{\perp}|^2 + (\frac{I^2R}{m^2}) \mp 2 \frac{I}{mR} \bm{v}_{\perp} \cdot \bm{x}_{\perp}
\end{gather}
arriving to :
\begin{gather}
I =  \pm 2m |\bm{v}_{\perp}|\, \cos(\,\theta( \bm{v}_{\perp},\bm{x}_{\perp} )\,),
\end{gather}
where $ \theta(\bm{v}_{\perp},\bm{x}_{\perp} )  $ is the angle between $\bm{v}_{\perp} $ and $ \bm{x}_{\perp}$;
note that $ I $  is always positive, indeed 
$\theta(\bm{v}_{\perp},\bm{x}_{\perp} ) \in (0,\frac{\pi}{2}) $ for inner collisions and
 $ \theta(\bm{v}_{\perp},\bm{x}_{\perp} ) \in (\frac{\pi}{2},\pi) $ for outer ones. 
 
Combining the expression of $ I $ with eq.(\ref{append vx tilde-vx =0}) and 
eq.(\ref{append m (vperp tilde -v)}) we obtain
\begin{subequations}
\begin{gather}
\widetilde{v_x}\,=\,v_x \\
\widetilde{\bm{v}}_{\perp} = \bm{v}_{\perp} - 2 |\bm{v}_{\perp}|\, \cos(\,\theta( \bm{v}_{\perp},\bm{x}_{\perp} ))\hat{\bm{x}}_{\perp}.
\end{gather}
\end{subequations}
\\
We now turn our attention to collisions with the base of our body, remembering that it is orthogonal to the x-axis.
In particular we consider that the particle hits the body in $ (x,\bm{x_{\perp} }) \in  D(t)  $ at time $ t $ and define its pre(post)collisional velocity as $ \bm{v} $ $  ( \bm{v}' )$ Proceeding similarly as before, conservation of momentum and  kinetic energy imply for the body:
\begin{equation}
V'=V+\frac{2m}{M+m}(v_x-V)   \\ 
\end{equation} 
where $V'$ and $v_x'$ are post-collisional velocities.\\
While for the particle:
\begin{gather}
v_x'=V-\frac{M-m}{M+m}(v_x-V) \\
\bm{v}_{\perp}' =\bm{v}_{\perp} 
\end{gather} 

Being $M>>m$, we have
\begin{gather}
V' \simeq V+\frac{2m}{M}(v_x-V)  
\end{gather}
and eq.(\ref{v'})
\begin{subequations}
\begin{gather}
v_x'\simeq 2V-v_x \\
\bm{v}_{\perp}' = \bm{v}_{\perp}.
\end{gather} 
\end{subequations}

\section{Appendix:   Properties of $ F_{\zeroped{0}}(V)$}
\label{appendix F_0(V)}
\begin{Lems}
$F_0(V)$ is an odd function and for $ V>0 $ it is positive, increasing and convex.
\end{Lems}

\begin{proof}
In the following $ V >0 $. By the change of variable $v_x \to -v_x$, 
\begin{align}
&F_0(V)=A\   \ \biggl [ \int _{-\infty} ^V dv_x (v_x -V)^2 e^{- \beta v_x ^2} -\int_{-\infty} ^{-V} dv_x (v_x +V)^2  e^{- \beta v_x ^2} \biggl ] \notag  \\
&= A\  \biggl [\int _{-\infty} ^{-V} dv_x (-4Vv_x) +\int _{-V} ^{V} dv_x (v_x -V)^2 e^{- \beta v_x ^2} \biggr ] \notag \\
&\geq A\  \int _{-\infty} ^{-V} dv_x (-4Vv_x) \geq 0 \notag
\end{align}

The first derivative: 
\begin{gather*} 
F_0'(V)= A\   \ \biggl [ \int _{-\infty} ^V dv_x  \ 2(V-v_x) e^{- \beta v_x ^2} - \int_{V} ^{\infty} dv_x \ 2(V-v_x)  e^{- \beta v_x ^2} \biggl ] = \\
2 A\   \ \biggl [ \int _{-\infty} ^V dv_x  \ (V-v_x) e^{- \beta v_x ^2} + \int_{V} ^{\infty} dv_x \ (v_x-V)  e^{- \beta v_x ^2} \biggl ]
\end{gather*}
is clearly positive .\\
The second derivative
\begin{equation}
F_0 ''(V)=2C \biggl [ \int _{-\infty} ^V dv_x e^{- \beta v_x ^2}- \int _{-\infty} ^{-V} dv_x e^{- \beta v_x ^2}  \biggl ] >0 \notag
\end{equation}
\\
In the end,
\begin{equation}
F_0(-V)=N   \ \biggl [ \int _{-\infty}^{-V} dv_x (v_x +V)^2 e^{- \beta v_x ^2} -\int_{- V}^{\infty} dv_x (v_x +V)^2  e^{- \beta v_x ^2} \biggr ] \notag
\end{equation}
By changing $v_x \to -v_x$ 
\begin{equation}
F_0(-V)=C   \ \biggl [\int_ V^{\infty} dv_x (v_x -V)^2  e^{- \beta v_x ^2}  -\int _{-\infty}^V dv_x (v_x -V)^2 e^{- \beta v_x ^2} \biggr ]=-F_0(V) \notag
\end{equation}
\end{proof}

\emph{Acknowledgments}  \,
We thank Carlo Marchioro and Guido Cavallaro for suggesting the problem and for their important observations and remarks during many discussions along the work.
\\
\
\\

\end{document}